\pdfoutput=1
\documentclass[%
reprint,
superscriptaddress,
nofootinbib,
 amsmath,amssymb,
 aps,
 prx,
 longbibliography,
floatfix,
]{revtex4-1}

\usepackage{graphicx}
\usepackage{dcolumn}
\usepackage{bm}
\usepackage{appendix}
\usepackage{hyperref}
\usepackage{graphicx}
\usepackage{amsmath}
\usepackage{blkarray, bigstrut}
\usepackage{amsthm}
\newtheorem{lemma}{Lemma}

\usepackage{mathtools}
\DeclarePairedDelimiter{\ceil}{\lceil}{\rceil}
\DeclarePairedDelimiter{\floor}{\lfloor}{\rfloor}
\DeclarePairedDelimiter{\abs}{\lvert}{\rvert}

\usepackage{array,amssymb,booktabs}
\newcolumntype{C}{>{$}c<{$}}
\AtBeginDocument{
    \heavyrulewidth=.08em
    \lightrulewidth=.05em
    \cmidrulewidth=.03em
    \belowrulesep=.65ex
    \belowbottomsep=0pt
    \aboverulesep=.4ex
    \abovetopsep=0pt
    \cmidrulesep=\doublerulesep
    \cmidrulekern=.5em
    \defaultaddspace=.5em
}

\usepackage[usenames, dvipsnames]{color}

\newcommand\mybigstrut[1][4pt]{\setlength\bigstrutjot{#1}\bigstrut[t]}
\newcommand\mynegbigstrut[1][2pt]{\setlength\bigstrutjot{#1}\bigstrut[b]}
\newcommand{\Mod}[1]{\ (\mathrm{mod}\ #1)}

\begin{document}

\preprint{APS/123-QED}

\title{Matrix optimization on universal unitary photonic devices}

\author{Sunil Pai}
\email{sunilpai@stanford.edu}
\affiliation{Department of Electrical Engineering, Stanford
University, Stanford, CA 94305, USA}
\author{Ben Bartlett}
\affiliation{Department of Applied Physics, Stanford University, Stanford, CA 94305, USA}
\author{Olav Solgaard}
\affiliation{Department of Electrical Engineering, Stanford
University, Stanford, CA 94305, USA}
\author{David A. B. Miller}
\email{dabm@stanford.edu}
\affiliation{Department of Electrical Engineering, Stanford
University, Stanford, CA 94305, USA}

\begin{abstract}
Universal unitary photonic devices can apply arbitrary unitary transformations to a vector of input modes and provide a promising hardware platform for fast and energy-efficient machine learning using light. We simulate the gradient-based optimization of random unitary matrices on universal photonic devices composed of imperfect tunable interferometers. If device components are initialized uniform-randomly, the locally-interacting nature of the mesh components biases the optimization search space towards banded unitary matrices, limiting convergence to random unitary matrices. We detail a procedure for initializing the device by sampling from the distribution of random unitary matrices and show that this greatly improves convergence speed. We also explore mesh architecture improvements such as adding extra tunable beamsplitters or permuting waveguide layers to further improve the training speed and scalability of these devices.
\end{abstract}

\maketitle


\section{\label{sec:intro}Introduction}

Universal multiport interferometers are optical networks that perform arbitrary unitary transformations on input vectors of coherent light modes. Such devices can be used in applications including quantum computing (e.g. boson sampling, photon walks) \cite{Grafe2016IntegratedWalks, Carolan2015UniversalOptics, Spring2013BosonChip, Harris2017QuantumProcessor}, mode unscramblers \cite{Annoni2017UnscramblingModes}, photonic neural networks \cite{Shen2017DeepCircuits, Hughes2018TrainingMeasurement, Williamson2019ReprogrammableNetworks}, and finding optimal channels through lossy scatterers \cite{Miller2013EstablishingAutomatically}. While universal photonic devices have been experimentally realized at a relatively small scale \cite{Shen2017DeepCircuits, Annoni2017UnscramblingModes}, commercial applications such as hardware for energy-efficient machine learning and signal processing can benefit from scaling the devices to up to $N = 1000$ modes. At this scale, fabrication imperfections and components with scale-dependent sensitivities can negatively affect performance.

One canonical universal photonic device is the rectangular multiport interferometer mesh \cite{Clements2016AnInterferometers} shown in Figure \ref{fig:rectangularmesh} interfering $N = 8$ modes. In multiport interferometers, an $N$-dimensional vector is represented by an array of modes arranged in $N$ single-mode waveguides. A unitary operation is applied to the input vector by tuning Mach-Zehnder interferometers (MZIs) represented by the red dots of Figure \ref{fig:rectangularmesh}. Each MZI is a two-port optical component made of two 50:50 beamsplitters and two tunable single-mode phase shifters. Other mesh architectures have been proposed, such as the triangular mesh \cite{Reck1994ExperimentalOperator} (shown in Appendix \ref{sec:haarmeasure}), the universal cascaded binary tree architecture \cite{Miller2013Self-aligningCoupler}, and lattice architectures where light does not move in a forward-only direction \cite{Perez2017HexagonalInterferometers, Perez2018SiliconCore, Perez2017MultipurposeCore}.

\begin{figure}[t]
    \centering
    \includegraphics[width=0.48\textwidth]{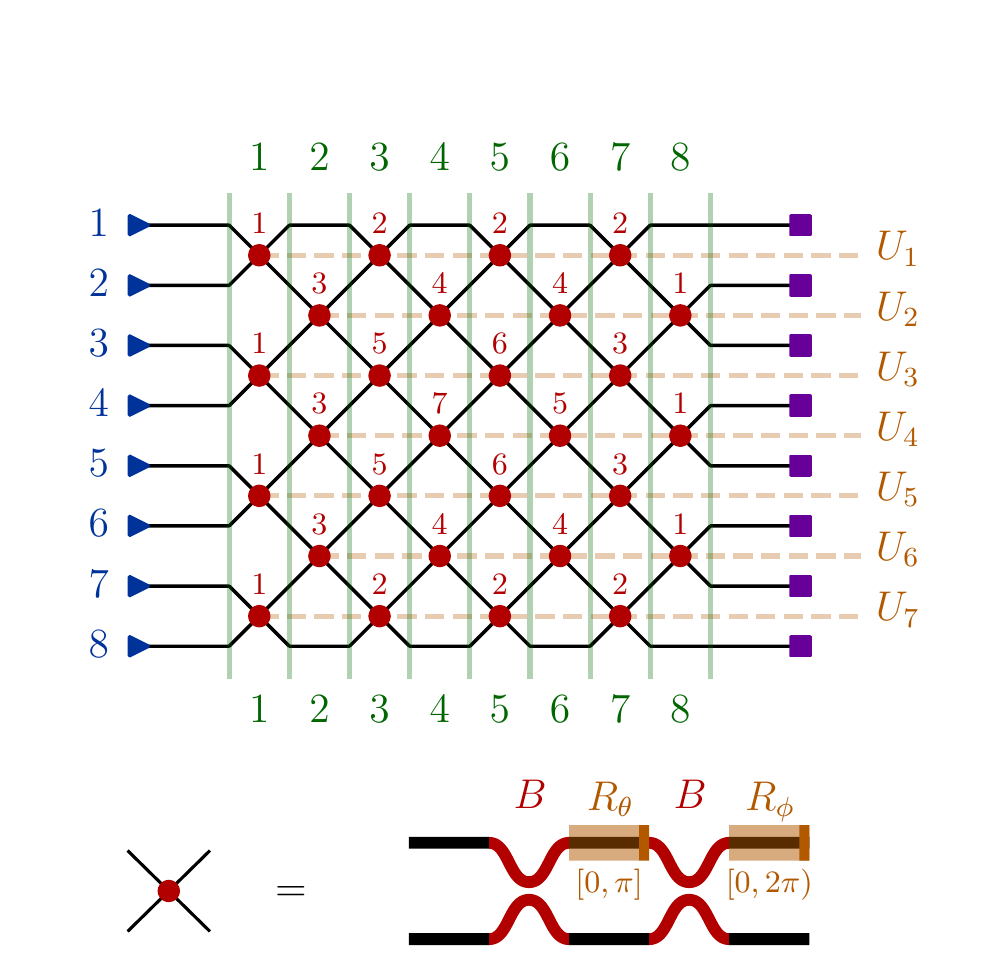}
    \caption{Mesh diagram representing the locally interacting rectangular mesh for $N = 8$. The inputs (and single-mode phase shifts at the inputs) are represented by blue triangles. Outputs are represented by purple squares. The MZI nodes are represented by red dots labelled with sensitivity index $\alpha_{n\ell}$ (e.g., $\alpha_{44} = 7$ is the most sensitive node). The nodes represent the Givens rotation $U_n$ (in orange) at vertical layer $\ell$ (in green). Each photonic MZI node can be represented with 50:50 beamsplitters $B$ (red) and phase shifters $R_\theta, R_\phi$ (orange) with required ranges $0 \leq \theta \leq \pi$ and $0 \leq \phi < 2\pi$.}
    \label{fig:rectangularmesh}
\end{figure}

The scalability of optimizing mesh architectures, especially using gradient-based methods, is limited by the ability of the locally interacting architecture to control the output powers in the mesh. If phase shifts in the mesh are initialized uniform-randomly, light propagates through the device in a manner similar to a random walk. The off-diagonal, nonlocal elements of the implemented unitary matrix tend to be close to zero because transitions between inputs and outputs that are far apart have fewer paths (e.g., input $1$ and output $8$ in Figure \ref{fig:rectangularmesh} have a single path). The resulting mesh therefore implements a unitary matrix with a banded structure that is increasingly pronounced as the matrix size increases.

In many applications such as machine learning \cite{Shen2017DeepCircuits} and quantum computing \cite{Russell2017DirectMatrices, Carolan2015UniversalOptics}, we avoid this banded unitary matrix behavior in favor of random unitary matrices. A random unitary matrix is achieved when the device phase shifts follow a distribution derived from random matrix theory \cite{Hurwitz1897UberIntegration, Zyczkowski1994RandomMatrices, Diaconis2017HurwitzMathematics, Spengler2012CompositeGroups, Russell2017DirectMatrices}. In the random matrix theory model, we assign a sensitivity index to each component that increases towards the center of the mesh, as shown in Figure \ref{fig:rectangularmesh}. The more sensitive components toward the center of the mesh require higher transmissivities and tighter optimization tolerances. If the required tolerances are not met, the implemented unitary matrix begins to show the undesired banded behavior.

In Section \ref{sec:meshintro}, we introduce the photonic mesh architecture and sources of error that can exacerbate the banded unitary matrix problem. In Section \ref{sec:haarinit}, we explicitly model the component settings to implement a random unitary matrix and ultimately avoid the banded unitary matrix problem. We propose a ``Haar initialization" procedure that allows light to propagate uniformly to all outputs from any input. We use this procedure to initialize the gradient-based optimization of a photonic mesh to learn unknown random unitary matrices given training data. We show that this optimization converges even in the presence of significant simulated fabrication errors.

In Sections \ref{sec:architectures} and \ref{sec:simulation}, we propose and simulate two alterations to the mesh architecture that further improve gradient-based optimization performance. First, we add redundant MZIs in the mesh to reduce convergence error by up to five orders of magnitude. Second, we permute the mesh interactions while maintaining the same number of tunable components, which increases allowable tolerances of phase shifters, decreases off-diagonal errors, and improves convergence time.

\section{\label{sec:meshintro} Photonic Mesh}

We define the photonic mesh when operated perfectly and then discuss how beam splitter or phase shift errors can affect device performance.

\subsection{Photonic unitary implementation}

A single-mode phase shifter can perform an arbitrary $\mathrm{U}(1)$ transformation $e^{i\phi}$ on its input. A phase-modulated Mach-Zehnder interferometer (MZI) with perfect ($50:50$) beamsplitters can apply to its inputs a unitary transformation $U$ of the form:
\begin{equation}\label{eqn:mzidefinition}
\begin{aligned}
U(\theta, \phi) &:= R_\phi B R_\theta B \\
&= \begin{bmatrix}e^{i\phi} & 0 \\ 0 & 1\end{bmatrix} \frac{1}{\sqrt{2}} \begin{bmatrix}1 & i \\ i & 1\end{bmatrix} \begin{bmatrix}e^{i\theta} & 0 \\ 0 & 1\end{bmatrix} \frac{1}{\sqrt{2}} \begin{bmatrix}1 & i \\ i & 1\end{bmatrix} \\
&= ie^{\frac{i\theta}{2}} \begin{bmatrix}e^{i\phi}\sin \frac{\theta}{2} & e^{i\phi}\cos \frac{\theta}{2} \\ \cos \frac{\theta}{2} & -\sin \frac{\theta}{2} \\
\end{bmatrix},
\end{aligned}
\end{equation}
where $B$ is the beamsplitter operator, $R_\theta, R_\phi$ are upper phase shift operators. Equation \ref{eqn:mzidefinition} is represented diagrammatically by the configuration in Figure \ref{fig:rectangularmesh}.\footnote{Other configurations with two independent phase shifters between the beamsplitters $B$ are ultimately equivalent for photonic meshes \cite{Miller2015PerfectComponents}.} If one or two single-mode phase shifters are added at the inputs, we can apply an arbitrary $\mathrm{SU}(2)$ or $\mathrm{U}(2)$ transformation to the inputs, respectively.

We define the transmissivity and reflectivity of the MZI as:
\begin{equation} \label{eqn:tr}
    \begin{aligned}
    t &:= \cos^2\left(\frac{\theta}{2}\right) = |U_{12}|^2 = |U_{21}|^2 \\
    r &:= \sin^2\left(\frac{\theta}{2}\right) = 1 - t = |U_{11}|^2 = |U_{22}|^2.
    \end{aligned}
\end{equation}
In this convention, when $\theta = \pi$, we have $r = 1, t = 0$ (the MZI ``bar state''), and when $\theta = 0$, we have $r = 0, t = 1$ (the MZI ``cross state'').

If there are $N$ input modes and the interferometer is connected to waveguides $n$ and $n + 1$ then we can embed the $2 \times 2$ unitary $U$ from Equation \ref{eqn:mzidefinition} in $N$-dimensional space with a locally-interacting unitary ``Givens rotation" $U_n$ defined as:
\begin{equation}\label{eqn:givensrotation}
U_{n} \,:=\,\, \begin{blockarray}{cccccccc}
 &  & \small{n} & \small{n+1} &  & & \\
\begin{block}{[cccccc]cl}
1   & \cdots &    0   &   0   & \cdots &    0 & & \mybigstrut\\
\vdots & \ddots & \vdots &  \vdots &        & \vdots & &\\
0   & \cdots &    U_{11}  &   U_{12}   & \cdots &    0 & & \small{n} \\
0   & \cdots &    U_{21}  &   U_{22}  & \cdots &    0   & & \small{n+1}\\
\vdots &   & \vdots  & \vdots & \ddots & \vdots & &\\
0   & \cdots &   0 &    0   & \cdots &    1 & & \mynegbigstrut\\
\end{block}
 &  & & &  & & \\
\end{blockarray}.\\
\end{equation}
All diagonal elements are 1 except those labeled $U_{11}$ and $U_{22}$, which have magnitudes of $\sqrt{r} = \sqrt{1-t}$, and all off-diagonal elements are 0 except those labeled $U_{12}$ and $U_{21}$, which have magnitudes of $\sqrt{t}$.

Arbitrary unitary transformations can be implemented on a photonic chip using only locally interacting MZIs \cite{Reck1994ExperimentalOperator}. In this paper, we focus on optimizing a rectangular mesh \cite{Clements2016AnInterferometers} of MZIs; however, our ideas can be extended to other universal schemes, such as the triangular mesh \cite{Miller2013Self-configuringInvited}, as well.

In the rectangular mesh scheme \cite{Clements2016AnInterferometers} of Figure \ref{fig:rectangularmesh}, we represent $\hat{U}_{\mathrm{R}} \in \mathrm{U}(N)$ in terms of $N(N - 1) / 2$ locally interacting Givens rotations $U_{n}$ and $N$ single-mode phase shifts at the inputs represented by diagonal unitary $D(\gamma_1, \gamma_2, \ldots \gamma_N)$:
\begin{equation}\label{eqn:rectangularmesh}
\begin{aligned}
\hat{U}_{\mathrm{R}} &:= \prod_{\ell = 1}^N \prod_{n \in \mathcal{S}_{\ell, N}} U_{n}(\theta_{n\ell}, \phi_{n\ell}) \cdot D(\gamma_1, \gamma_2, \ldots \gamma_N),
\end{aligned}
\end{equation}
where our layer-wise product left-multiplies from $\ell = N$ to 1,\footnote{In general, for matrix products for a sequence $\{M_\ell\}$, we define the multiplication order $\prod_{\ell = 1}^N M_\ell = M_N M_{N-1} \cdots M_1$.} the single-mode phase shifts are $\gamma_n \in [0, 2 \pi)$, and where the Givens rotations are parameterized by $\theta_{n\ell} \in [0, \pi], \phi_{n\ell} \in [0, 2\pi)$.\footnote{Since $\gamma_n, \phi_{n\ell}$ are periodic phase parameters, they are in half-open intervals $[0, 2\pi)$. In contrast, any $\theta_{n\ell} \in [0, \pi]$ must be in a closed interval to achieve all transmissivities $t_{n\ell} \in [0, 1]$.} We define the top indices of each interacting mode for each vertical layer as the set $\mathcal{S}_{\ell, N} = \{n \in [1, 2, \ldots N-1] \mid n \Mod 2 \equiv \ell \Mod 2\}$. This vertical layer definition follows the convention of Refs. \cite{Jing2017TunableRNNs, Hughes2018TrainingMeasurement} and is depicted in Figure \ref{fig:rectangularmesh}, where $\ell$ represents the index of the vertical layer.

\subsection{Beamsplitter error tolerances}
The expressions in Equations \ref{eqn:mzidefinition} and \ref{eqn:rectangularmesh} assume perfect fabrication. In practice, however, we would like to simulate how practical devices with errors in each transfer matrix $B, R_\phi, R_\theta$ in Equation \ref{eqn:mzidefinition} impact optimization performance.

In fabricated chip technologies, imperfect beamsplitters $B$ can have a split ratio error $\epsilon$ that change the behavior of the red 50:50 coupling regions in Figure \ref{fig:rectangularmesh} or $B$ in Equation \ref{eqn:mzidefinition}. The resultant scattering matrix $U_\epsilon$ with imperfect beamsplitters $B_\epsilon$ can be written as:
\begin{equation}\label{eqn:beamsplittererror}
\begin{aligned}
B_\epsilon &:= \frac{1}{\sqrt{2}}\begin{bmatrix}\sqrt{1 + \epsilon} & i\sqrt{1 - \epsilon} \\ i\sqrt{1 - \epsilon} & \sqrt{1 + \epsilon}\end{bmatrix}\\
U_\epsilon &:= R_\phi B_\epsilon R_\theta B_\epsilon.
\end{aligned}
\end{equation}
As shown in Appendix \ref{sec:photonicerrorsproof}, if we assume both beamsplitters have identical $\epsilon$, we find $t_\epsilon := t (1 - \epsilon^2) \in [0, 1-\epsilon^2]$ is the realistic transmissivity, $r_\epsilon := r + t \cdot \epsilon^2 \in [\epsilon^2, 1]$ is the realistic reflectivity, and $t, r$ are the ideal transmissivity and reflectivity defined in Equation \ref{eqn:tr}.

The unitary matrices in Equation \ref{eqn:beamsplittererror} cannot express the full transmissivity range of the MZI, with errors of up to $\epsilon^2$ in the transmissivity, potentially limiting the performance of greedy progressive photonic algorithms \cite{Burgwal2017UsingUnitaries, Miller2017SettingMethod, Flamini2017BenchmarkingProcessing}. Our Haar phase theory, which we develop in the following section, determines acceptable interferometer tolerances for calibration of a ``perfect mesh" consisting of imperfect beamsplitters \cite{Miller2015PerfectComponents} given large $N$. We will additionally show that simulated photonic backpropagation \cite{Hughes2018TrainingMeasurement} with adaptive learning can adjust to nearly match the performance of perfect meshes with errors as high as $\epsilon = 0.1$ for meshes of size $N = 128$.

\subsection{Phase shift tolerances}

Another source of uncertainty in photonic meshes is the phase shift tolerances of the mesh which affect the matrices $R_\theta, R_\phi$ of Equation \ref{eqn:mzidefinition}, shown in orange in Figure \ref{fig:rectangularmesh}. Error sources such as thermal crosstalk or environmental drift may result in slight deviance of phase shifts in the mesh from intended operation. Such errors primarily affect the control parameters $\theta_{n\ell}$ that control light propagation in the mesh by affecting the MZI split ratios. This nontrivial problem warrants a discussion of mean behavior and sensitivities (i.e., the distribution) of $\theta_{n\ell}$ needed to optimize a random unitary matrix.

\section{Haar Initialization} \label{sec:haarinit}

\subsection{Cross state bias and sensitivity index}

The convergence of global optimization depends critically on the sensitivity of each phase shift. The gradient descent optimization we study in this paper converges when the phase shifts are correct to within some acceptable range. This acceptable range can be rigorously defined in terms of average value and variance of phase shifts in the mesh that together define an unbiased (``Haar random") unitary matrix.\footnote{A Haar random unitary is defined as Gram-Schmidt orthogonalization of $N$ standard normal complex vectors \cite{Spengler2012CompositeGroups, Russell2017DirectMatrices}.} To implement a Haar random unitary, some MZIs in the mesh need to be biased towards cross state ($t_{n\ell}$ near $1$, $\theta_{n\ell}$ near $0$) \cite{Burgwal2017UsingUnitaries, Russell2017DirectMatrices}. This cross state bias correspondingly ``pinches" the acceptable range for transmissivity and phase shift near the limiting cross state configuration, resulting in higher sensitivity, as can be seen in Figure \ref{fig:haarphase}(b).

For an implemented Haar random unitary matrix, low-tolerance, transmissive MZIs are located  towards the center of a rectangular mesh \cite{Russell2017DirectMatrices, Burgwal2017UsingUnitaries} and the apex of a triangular mesh as proven in Appendix \ref{sec:haarmeasure}. For both the triangular and rectangular meshes, the cross state bias and corresponding sensitivity for each MZI depends only on the total number of reachable waveguides ports, as proven in Appendix \ref{sec:haarproof}. Based on this proof, we define the sensitivity index $\alpha_{n\ell} := \abs{I_{n\ell}} + \abs{O_{n\ell}} - N - 1$,\footnote{Note that $1 \leq \alpha_{n\ell} \leq N - 1$, and there are always $N - \alpha_{n\ell}$ MZIs that have a sensitivity index of $\alpha_{n\ell}$.} where $I_{n\ell}$ and $O_{n\ell}$ are the subsets of input and output waveguides reachable by light exiting or entering the MZI, respectively, and $\abs{\cdot}$ denotes set size. Figure \ref{fig:rectangularmesh} and Figure \ref{fig:checkerboardplots}(a) show the sensitivity index for the rectangular mesh, which clearly increases towards the center MZI.

\begin{figure}[h]
    \centering
    \includegraphics[width=0.48\textwidth]{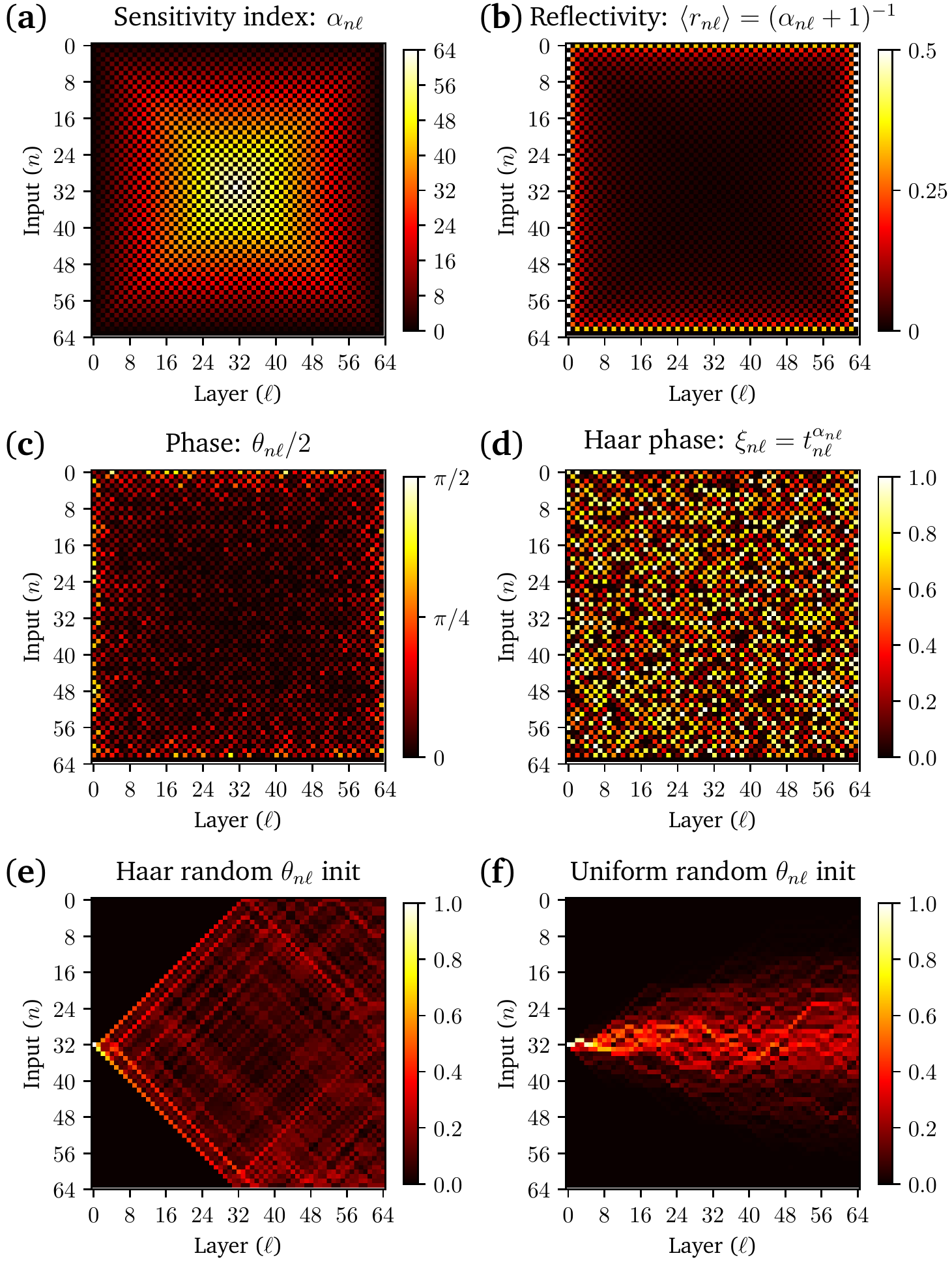}
    \caption{(a) The sensitivity index $\alpha_{n\ell}$ for $N = 64$. (b) Checkerboard plot for the average reflectivity $\left\langle r_{n\ell} \right\rangle$ in a rectangular mesh. (c) Haar-random matrix and run the decomposition in Ref. \cite{Clements2016AnInterferometers} to find phases approaching cross state in the middle of the mesh. (d) The Haar phase $\xi_{n\ell}$ for the rectangular mesh better displays the randomness. (e, f) Field measurements (absolute value) from propagation at input 32 in (e) Haar and (f) uniform random initialized rectangular meshes with $N = 64$.}
    \label{fig:checkerboardplots}
\end{figure}

\subsection{Phase shift distributions and Haar phase}

The external $\phi_{n\ell}, \gamma_n$ phase shifts do not affect the the transmissivity $t_{n\ell}$ and therefore obey uniform random distributions \cite{Russell2017DirectMatrices}. In contrast, the $\theta_{n\ell}$ phase shifts have a probability density function (PDF) that depends on $\alpha_{n\ell}$ \cite{Russell2017DirectMatrices}: 
\begin{equation}\label{eqn:haarpdf}
    \begin{aligned}
    \mathcal{P}_{\alpha_{n\ell}}\left(\frac{\theta_{n\ell}}{2}\right) &= \alpha_{n\ell}\sin\left(\frac{\theta_{n\ell}}{2}\right) \left[\cos\left(\frac{\theta_{n\ell}}{2}\right)\right]^{2\alpha_{n\ell} - 1}.
    \end{aligned}
\end{equation}

The general shape of this distribution is presented in Figure \ref{fig:haarphase}(b), showing how an increase in $\alpha_{n\ell}$ biases $\theta_{n\ell}$ towards the cross state with higher sensitivity.

We define the Haar phase $\xi_{n\ell}$ as the cumulative distribution function (CDF) of $\theta_{n\ell} / 2$ starting from $\theta_{n\ell} / 2 = \pi / 2$:
\begin{equation} \label{eqn:haarphasecdf}
    \xi_{n\ell} := \int_{\pi / 2}^{\theta_{n\ell} / 2} \mathcal{P}_{\alpha_{n\ell}}(\theta) \mathrm{d}\theta.
\end{equation}

Using Equations \ref{eqn:haarpdf} and \ref{eqn:haarphasecdf}, we can define $\xi_{n\ell}(\theta_{n\ell}) \in [0, 1]$ that yields a Haar random matrix:
\begin{equation}\label{eqn:haarphase}
\begin{aligned}
    \xi_{n\ell} &= \left[\cos^2\left(\frac{\theta_{n\ell}}{2}\right)\right]^{\alpha_{n\ell}} = t_{n\ell}^{\alpha_{n\ell}},
\end{aligned}
\end{equation}
where $t_{n\ell}$ represents the transmissivity of the MZI, which is a function of $\theta_{n\ell}$ as defined in Equation \ref{eqn:tr}. 

\subsection{Haar initialization}

In the physical setting, it is useful to find the inverse of Equation \ref{eqn:haarphase} to directly set the measurable transmissivity $t_{n\ell}$ of each MZI using a uniformly varying Haar phase $\xi_{n\ell} \sim \mathcal{U}(0, 1)$, a process we call ``Haar initialization" shown in Figure \ref{fig:checkerboardplots}(c, d):
\begin{equation}\label{eqn:haarinit}
\begin{aligned}
    t_{n\ell} &= \sqrt[\leftroot{-2}\uproot{5}\alpha_{n\ell}]{\xi_{n\ell}} \\
    \theta_{n\ell} &= 2 \arccos \sqrt{t_{n\ell}} = 2 \arccos \sqrt[\leftroot{-2}\uproot{5}2\alpha_{n\ell}]{\xi_{n\ell}},
\end{aligned}
\end{equation}
where the expression for $\theta_{n\ell}$ is just a rearrangement of Equation \ref{eqn:tr}.

Haar initialization can be achieved progressively using a procedure similar to that in Ref. \cite{Miller2017SettingMethod}. If the phase shifters in the mesh are all well-characterized, the transmissivities can be directly set \cite{Russell2017DirectMatrices}. We will show in Section \ref{sec:simulation} that Haar initialization improves the convergence speed of gradient descent optimization significantly.

We can also use Equation \ref{eqn:haarinit} to find the average transmissivity and reflectivity for an MZI parameterized by $\alpha_{n\ell}$ as is found through simulation in Ref. \cite{Burgwal2017UsingUnitaries}:
\begin{equation} \label{eqn:avgbehavior}
    \begin{aligned}
    \left\langle t_{n\ell} \right\rangle &= \int_0^1 d \xi_{n\ell} \sqrt[\leftroot{-2}\uproot{5}\alpha_{n\ell}]{\xi_{n\ell}} = \frac{\alpha_{n\ell}}{\alpha_{n\ell} + 1}\\
    \left\langle r_{n\ell}\right\rangle &=  \frac{1}{\alpha_{n\ell} + 1} = \frac{1}{\abs{I_{n\ell}} + \abs{O_{n\ell}} - N}.
    \end{aligned}
\end{equation}
The average reflectivity $\langle r_{n\ell} \rangle$ shown in Figure \ref{fig:checkerboardplots}(b) gives a simple interpretation for the sensitivity index shown in Figure \ref{fig:checkerboardplots}(a). The average reflectivity is equal to the inverse of the total number of inputs and outputs reachable by the MZI minus the number of ports on either side of the device, $N$. This is true regardless of whether $\alpha_{n\ell}$ is assigned for a triangular or rectangular mesh.

\begin{figure}[t]
    \centering
    \includegraphics[width=0.48\textwidth]{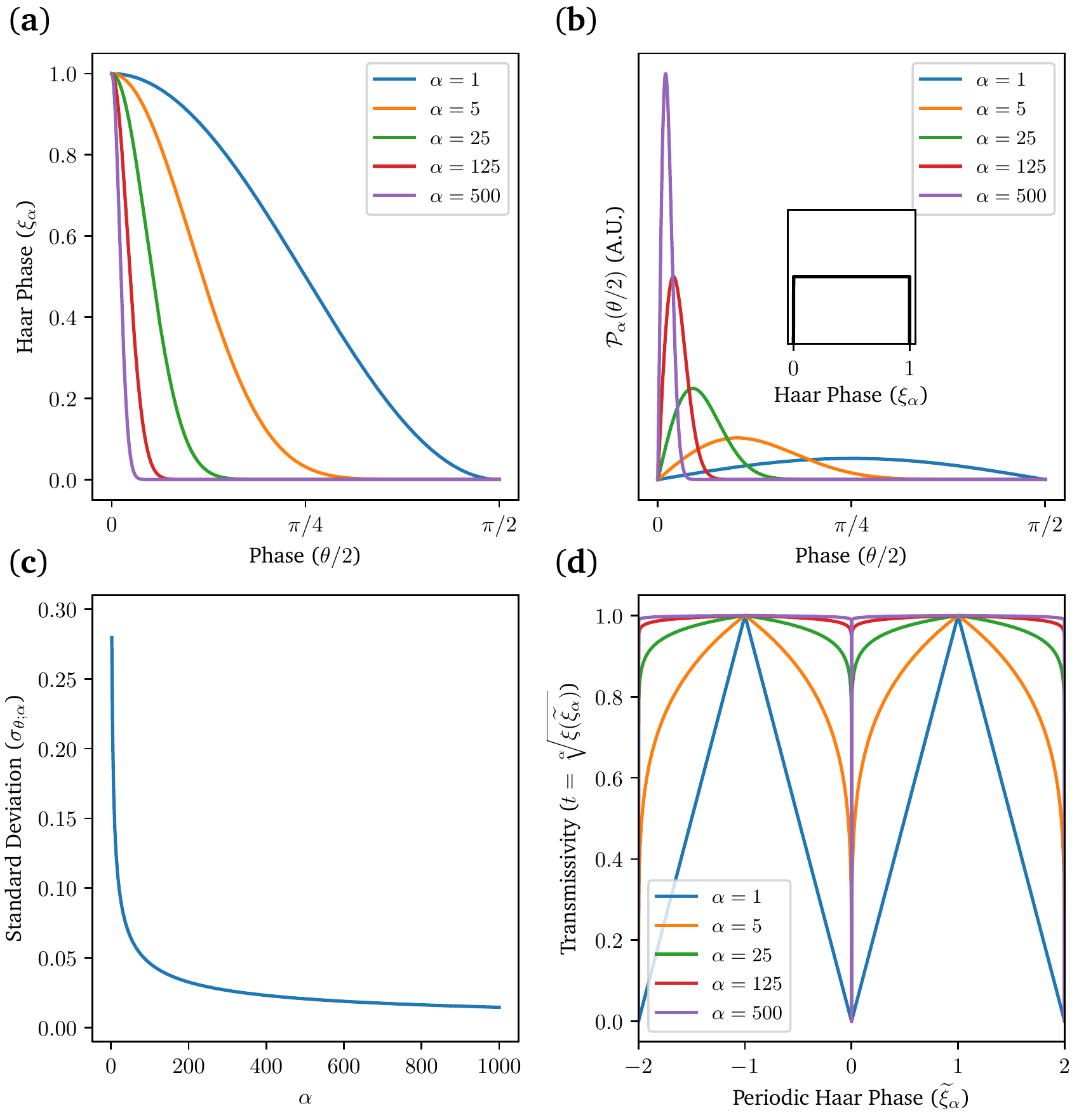}
    \caption{(a) Plot of the relationship between $\xi_\alpha$ and $\theta$. (b) We show that phase shift standard deviation $\sigma_{\theta; \alpha}$ decreases as $\alpha$ increases. (c) A plot of $\sigma_{\theta; \alpha}$ as $\alpha$ increases. (d) The transmissivity of an MZI component as a function of a periodic Haar phase has a power law relationship. The periodic Haar phase $\widetilde{\xi}_\alpha$ is mapped to the Haar phase by a function $\xi: \mathbb{R} \to [0, 1]$ as discussed in Appendix \ref{sec:periodicparameters}.}
    \label{fig:haarphase}
\end{figure}

To see what the Haar initialization has accomplished, we can compare the field propagation through the rectangular mesh from a single input when Haar initialized versus uniform initialized in Figure \ref{fig:checkerboardplots}(e). Physically, this corresponds to light in the mesh spreading out quickly from the input of the mesh and ``interacting" more near the boundaries of the mesh (inputs, outputs, top, and bottom), as compared to the center of the mesh which has high transmissivity. In contrast, when phases are randomly set, the light effectively follows a random walk through the mesh, resulting in the field propagation pattern shown in Figure \ref{fig:checkerboardplots}(f).

\subsection{Tolerance dependence on $N$}

While Haar initialization is based on how the average component reflectivity scales with $N$, optimization convergence and device robustness ultimately depend on how phase shift tolerances scale with $N$. The average sensitivity index in the mesh is $\langle\alpha_{n\ell}\rangle = (N+1)/3$. As shown in Figure \ref{fig:haarphase}(b, c), the standard deviation $\sigma_{\theta; \alpha}$ over the PDF $\mathcal{P}_\alpha$ decreases as $\alpha$ increases. Therefore, a phase shifter's allowable tolerance, which roughly correlates with $\sigma_{\theta; \alpha}$, decreases as the total number of input and output ports affected by that component increases. Since $\langle\alpha_{n\ell}\rangle$ increases linearly with $N$, the required tolerance gets more restrictive at large $N$, as shown in Figure \ref{fig:haarphase}(c). We find that the standard deviation is on the order $10^{-2}$ radians for most values of $N$ in the specified range. Thus, if thermal crosstalk is ignored \cite{Shen2017DeepCircuits}, it is possible to implement a known random unitary matrix in a photonic mesh assuming perfect operation. However, we concern ourselves with on-chip optimization given just input/output data, in which case the unitary matrix is unknown. In such a case, the decreasing tolerances do pose a challenge in converging to a global optimum as $N$ increases. We demonstrate this problem for $N = 128$ in Section \ref{sec:simulation}.

To account for the scalability problem in global optimization, one strategy may be to design a component in such a way that the mesh MZIs can be controlled by Haar phase voltages as in Figure \ref{fig:haarphase}(d) and Equation \ref{eqn:haarinit}. The transmissivity dependence on a periodic Haar phase (shown in Figure \ref{fig:haarphase}(d) and discussed in Appendix \ref{sec:periodicparameters}), is markedly different from the usual sinusoidal dependence on periodic $\theta_{n\ell}$. The MZIs near the boundary vary in transmissivity over a larger voltage region than the MZIs near the center, where only small voltages are needed get to full transmissivity. This results in an effectively small control tolerance near small voltages. This motivates the modifications to the mesh architecture which we discuss in the next section.

\section{Architecture Modifications}
\label{sec:architectures}

We propose two architecture modifications that can relax the transmissivity tolerances in the mesh discussed in Section \ref{sec:haarinit} and result in significant improvement in optimization.

\subsection{Redundant rectangular mesh (RRM)}

By adding extra tunable MZIs, it is possible to greatly accelerate the optimization of a rectangular mesh to an unknown unitary matrix. The addition of redundant tunable layers to a redundant rectangular mesh is depicted in green in Figure \ref{fig:architectures}(a). The authors in Ref. \cite{Burgwal2017UsingUnitaries} point out that using such ``underdetermined meshes" (number of inputs less than the number of tunable layers in the mesh) can overcome photonic errors and restore fidelity in unitary construction algorithms. Adding layers to the mesh increases the overall optical depth of the device, but embedding smaller meshes with extra beamsplitter layers in a rectangular mesh of an acceptable optical depth does not pose intrinsic waveguide loss-related problems.

\subsection{Permuting rectangular mesh (PRM)}

Another method to accelerate the optimization of a rectangular mesh is to shuffle outputs at regular intervals within the rectangular mesh. This shuffling relaxes component tolerances and uniformity of the number of paths for each input-output transition. We use this intuition to formally define a permuting rectangular mesh. For simplicity,\footnote{If $N$ is not a power of 2, then one might consider the following approximate design: $K = \ceil{\log_2 N}$. Define $b(K) = \sqrt[K]{N}$, and let each $P_k$ have $\ceil{b^k}$ layers.} assume $N = 2^{K}$ for some positive integer $K$. Define ``rectangular permutation" operations $P_k$ that allow inputs to interact with waveguides at most $2^k$ away for $k < K$. These rectangular permutation blocks can be implemented using a rectangular mesh composed of MZIs with fixed cross state phase shifts, as shown in Figure \ref{fig:architectures}(b), or using low-loss waveguide crossings.

We now add permutation matrices $P_1, P_2, \ldots P_{K-1}$ into the middle of the rectangular mesh as follows
\begin{equation}\label{eqn:permutingrectangularmesh}
\begin{aligned}
    \hat{U}_{\mathrm{PR}} &:= M_K \left(\prod_{k=1}^{K - 1} P_k M_k  \right) \\ 
    M_k &:= \prod_{\ell=(k-1) \ceil{\frac{N}{K}}}^{\min\left(k\ceil{\frac{N}{K}}, N\right)} \prod_{n \in \mathcal{S}_{\ell, N}} U_n(\theta_{n\ell}, \phi_{n\ell}),
\end{aligned}
\end{equation}
where $\ceil{x}$ represents the nearest integer larger than $x$.

There are two operations per block $k$: an $\ceil{\frac{N}{K}}$-layer rectangular mesh which we abbreviate as $M_k$, and the rectangular permutation mesh $P_k$ where block index $k \in [1 \cdots K - 1]$. This is labelled in Figure \ref{fig:architectures}(b).

\begin{figure}[t]
    \centering
    \includegraphics[width=0.48\textwidth, trim={0 3.5cm 0 0},clip]{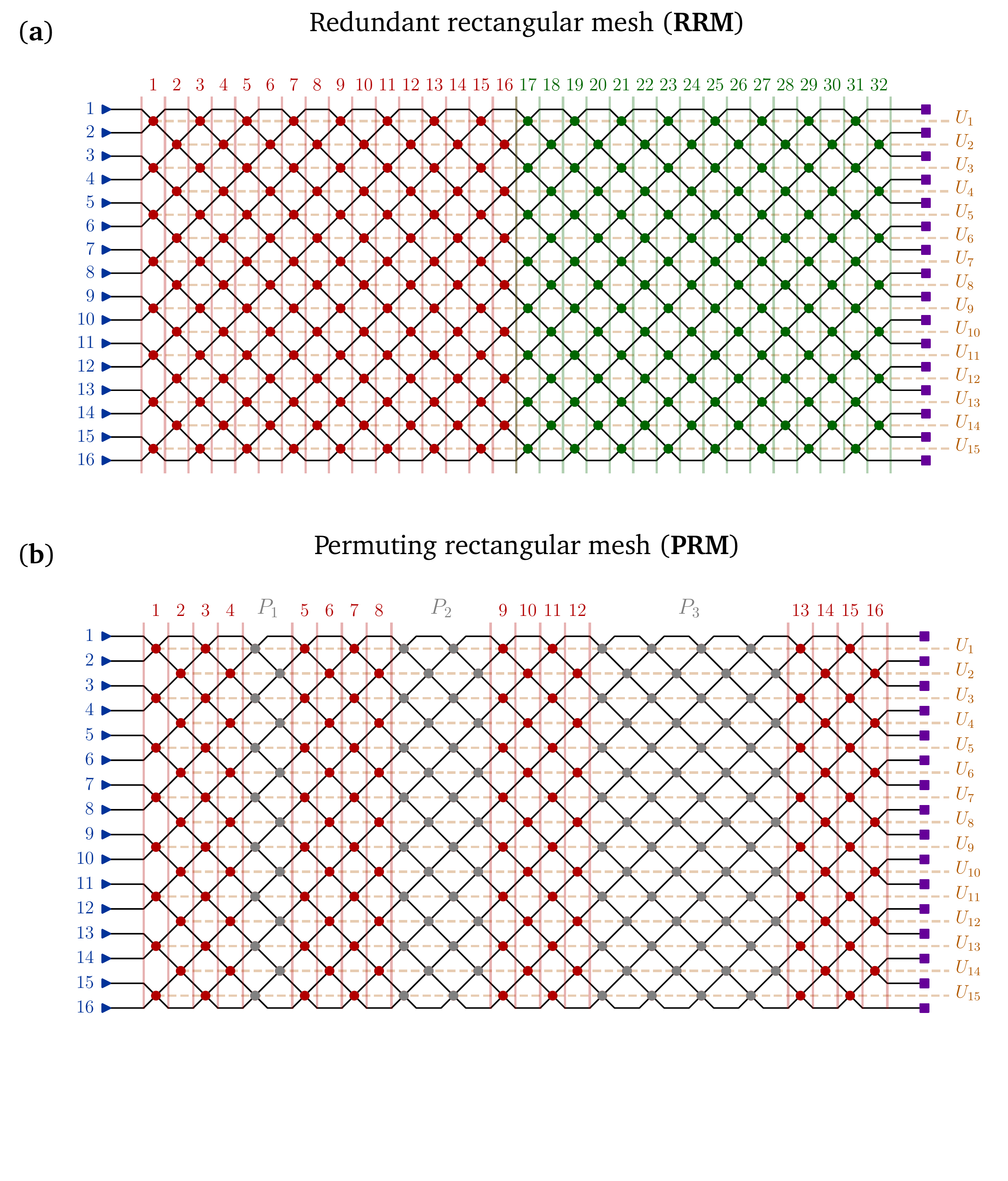}
    \caption{(a) A $16 \times 16$ rectangular mesh (red). Extra tunable layers (green) may be added to significantly reduce convergence time. (b) A $16$-input, $30$-layer permuting rectangular mesh. The rectangular permutation layer is implemented using either waveguide crossings or cross state MZIs (gray).}
    \label{fig:architectures}
\end{figure}

\section{Simulations \label{sec:simulation}}

Now that we have discussed the mesh modifications and Haar initialization, we simulate global optimization to show how our framework can improve convergence performance by up to five orders of magnitude, even in the presence of fabrication error.

\subsection{Mesh initialization}
\begin{figure}[t]
    \centering
    \includegraphics[width=0.48\textwidth]{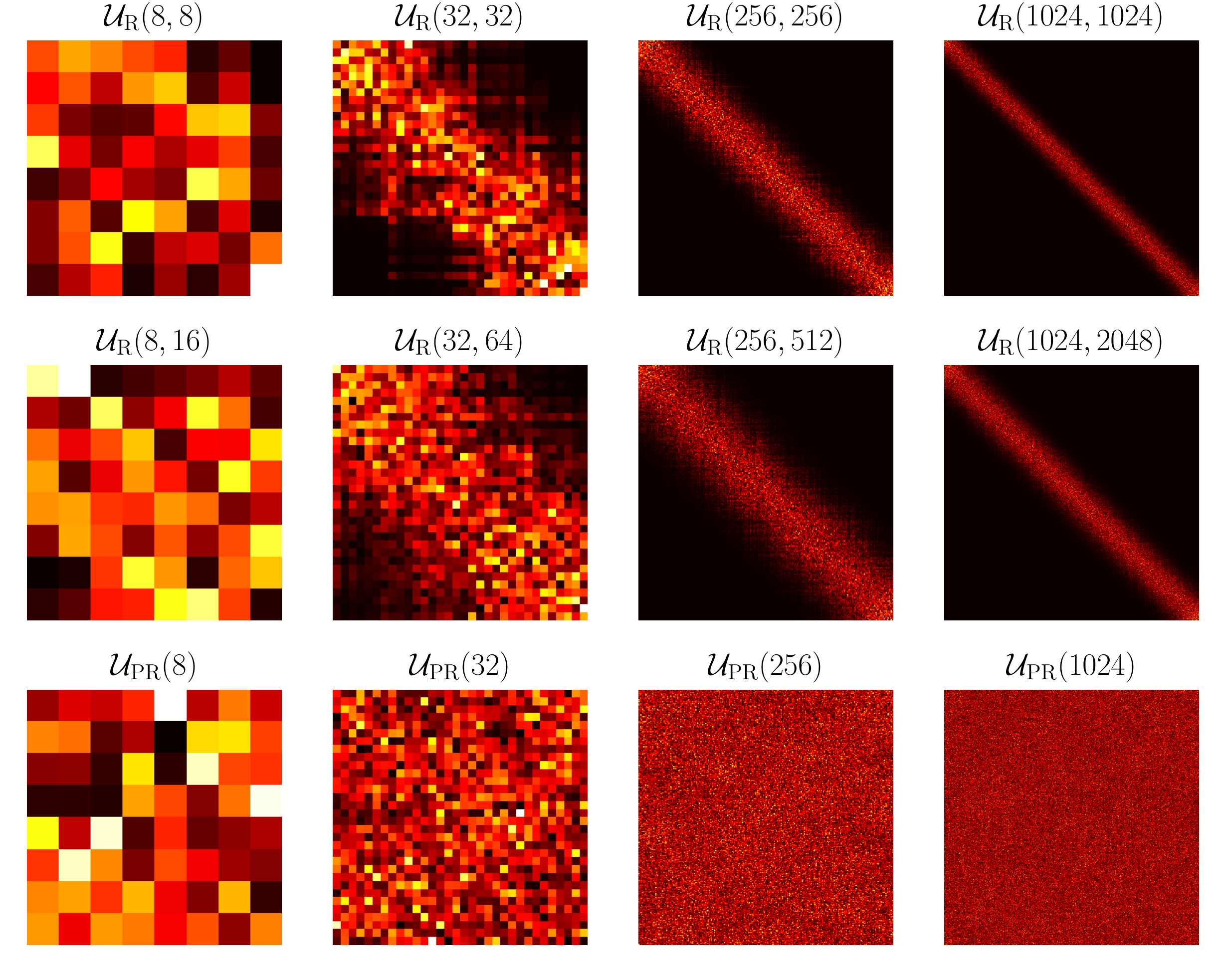}
    \caption{Elementwise absolute values of unitary matrices resulting from rectangular ($U \sim \mathcal{U}_{\mathrm{R}}$) and permuting rectangular ($U \sim \mathcal{U}_{\mathrm{PR}}$) meshes where meshes are initialized with uniform-random phases.}
    \label{fig:bandedmatrices}
\end{figure}

We begin by discussing the importance of initializing the mesh to respect the cross state bias and sensitivity of each component for Haar random unitary matrices discussed in Section \ref{sec:haarinit}. Uniform random phase initialization is problematic because it is agnostic of the sensitivity and average behavior of each component. We define this distribution of matrices as $\mathcal{U}_\mathrm{R}(N, L)$ for a rectangular mesh for $N$ inputs and $L$ layers. As shown previously in Figure \ref{fig:checkerboardplots}(f), any given input follows a random walk-like propagation if phases are initialized uniform-randomly, so there will only be non-zero matrix elements within a ``bandsize" about the diagonal. This bandsize decreases as circuit size $N$ increases as shown in Figure \ref{fig:bandedmatrices}.

We compare the bandsizes of banded unitary matrices in simulations qualitatively as we do in Figure \ref{fig:bandedmatrices} or quantitatively as we do in Appendix \ref{sec:bandsize}. We randomly generate $U \sim \mathcal{U}_\mathrm{R}(N, N)$, $U \sim \mathcal{U}_{\mathrm{PR}}(N)$ (permuting rectangular mesh with $N$ tunable layers), and $U \sim \mathcal{U}_{\mathrm{R}}(N, N + \delta N)$ (redundant rectangular mesh with $\delta N$ extra tunable layers). Figure \ref{fig:bandedmatrices} shows a significant reduction in bandsize as $N$ grows larger for rectangular meshes. This phenomenon is not observed with permuting rectangular meshes which generally have the same bandsize as Haar random matrices (independent of $N$) as shown in in Figure \ref{fig:bandedmatrices} and Appendix \ref{sec:bandsize}. This correlates with permuting rectangular meshes having faster optimization and less dependence on initialization.

Instead of initializing the mesh using uniform random phases, we use Haar initialization as in Equation \ref{eqn:haarinit} to avoid starting with a banded unitary configuration. This initialization, which we recommend for any photonic mesh-based neural network application, dramatically improves convergence because it primes the optimization with the right average behavior for each component. We find in our simulations that as long as the initialization is calibrated towards higher transmissivity ($\theta_{n\ell}$ near $0$), larger mesh networks can also have reasonable convergence times similar to when the phases are Haar-initialized. 

The proper initialization of permuting rectangular meshes is less clear because the tolerances and average behavior of each component have not yet been modeled. Our proposal is to initialize each tunable block $M_k$ as an independent mesh using the same definition for $\alpha_{n\ell}$, except replacing $N$ with the number of layers in $M_k$, $\lceil N / K\rceil$. This is what we use as the Haar initialization equivalent in the permuting rectangular mesh case, although it is possible there may be better initialization strategies for the nonlocal mesh structure.

\subsection{Optimization problem and synthetic data}

After initializing the photonic mesh, we proceed to optimize the mean square error cost function for an unknown Haar random unitary $U$:
\begin{equation} \label{eqn:optimization}
\begin{aligned}
& \underset{\theta_{n\ell}, \phi_{n\ell}, \gamma_n}{\text{minimize}}
& & \frac{1}{2N}\left\lVert\hat U(\theta_{n\ell}, \phi_{n\ell}, \gamma_n) - U\right\rVert_F^2,
\end{aligned}
\end{equation}
where the estimated unitary matrix function $\hat U$ maps $N^2$ phase shift parameters $\theta_{n\ell}, \phi_{n\ell}, \gamma_n$ to $\mathrm{U}(N)$ via Equations \ref{eqn:rectangularmesh} or \ref{eqn:permutingrectangularmesh}, and $\|\cdot\|_F$ denotes the Frobenius norm. Since trigonometric functions parameterizing $\hat U$ are non-convex, we know that Equation \ref{eqn:optimization} is a non-convex problem. The non-convexity of Equation \ref{eqn:optimization} suggests learning a single unitary transformation in a deep neural network might have significant dependence on initialization.

To train the network, we generate random unit-norm complex input vectors of size $N$ and generate corresponding labels by multiplying them by the target matrix $U$. We use a training batch size of $2N$. The synthetic training data of unit-norm complex vectors is therefore represented by $X \in \mathbb{C}^{N \times 2N}$. The minibatch training cost function is similar to the test cost function, $\mathcal{L}_\mathrm{train} = \|\hat U X - U X\|_F^2$. The test set is the identity matrix $I$ of size $N \times N$. The test cost function, in accordance with the training cost function definition, thus matches Equation \ref{eqn:optimization}.

\subsection{Training algorithm}

We simulate the global optimization of a unitary mesh using automatic differentiation in \texttt{tensorflow}, which can be physically realized using the \textit{in situ} backpropagation procedure in Ref. \cite{Hughes2018TrainingMeasurement}. This optical backpropagation procedure physically measures $\partial\mathcal{L}_\mathrm{train}/\partial\theta_{n\ell}$ using interferometric techniques, which can be extended to any of the architectures we discuss in this paper.

The on-chip backpropagation approach is also likely faster for gradient computation than other training approaches such as the finite difference method mentioned in past on-chip training proposals \cite{Shen2017DeepCircuits}. We find empirically that the Adam update rule (a popular first-order adaptive update rule \cite{Kingma2015Adam:Optimization}) outperforms standard stochastic gradient descent for the training of unitary networks. If gradient measurements for the phase shifts are stored during training, adaptive update rules can be applied using successive gradient measurements for each tunable component in the mesh. Such a procedure requires minimal computation (i.e., locally storing the previous gradient step) and can act as a physical test of the simulations we will now discuss. Furthermore, we avoid quasi-Newton optimization methods such as L-BFGS used in Ref. \cite{Burgwal2017UsingUnitaries} that cannot be implemented physically as straightforwardly as first-order methods.

The models were trained using our open source simulation framework \href{https://github.com/solgaardlab/neurophox}{\texttt{neurophox}} \footnote{See \href{https://github.com/solgaardlab/neurophox}{https://github.com/solgaardlab/neurophox}.} using a more general version of the vertical layer definition proposed in Refs. \cite{Jing2017TunableRNNs, Hughes2018TrainingMeasurement}. The models were programmed in \texttt{tensorflow} \cite{Abadi2016TensorFlow:Learning} and run on an NVIDIA GeForce GTX1080 GPU to improve optimization performance.

\subsection{Results}

We now compare training results for rectangular, redundant rectangular, and permuting rectangular meshes given $N = 128$. In our comparison of permuting rectangular meshes and rectangular meshes, we analyze performance when beamsplitter errors are distributed throughout the mesh as either $\epsilon = 0$ or $\epsilon \sim \mathcal{N}(0, 0.01)$ and when the $\theta_{n\ell}$ are randomly or Haar-initialized (according to the PDF in Equation \ref{eqn:haarpdf}). We also analyze optimization perforamnces of redundant rectangular meshes where we vary the number of vertical MZI layers.

From our results, we report five key findings:

\begin{enumerate}
    \item Optimization of $N = 128$ rectangular meshes results in significant off-diagonal errors due to bias towards the banded matrix space of $\mathcal{U}_R(128)$, as shown in Figure \ref{fig:unitaryopt}.
    \item Rectangular meshes converge faster when Haar-initialized than when uniformly random initialized, as in Figure \ref{fig:unitaryopt}, in which case the estimated matrix converges towards a banded configuration shown in Appendix \ref{sec:unitaryoptcomparison}.
    \item Permuting rectangular meshes converge faster than rectangular meshes despite having the same number of total parameters as shown in Figure \ref{fig:unitaryopt}.
    \item Redundant rectangular meshes, due to increase in the number of parameters, have up to five orders of magnitude better convergence when the number of vertical layers are doubled compared to rectangular and permuting rectangular meshes, as shown in Figure \ref{fig:rrm}.
    \item Beamsplitter imperfections slightly reduce the overall optimization performance of permuting and redundant rectangular meshes, but reduce the performance of the rectangular mesh significantly. (See Figure \ref{fig:unitaryopt} and Appendix \ref{sec:rrmphotonicerrors}.)
\end{enumerate}

\begin{figure}[h]
    \centering
    \includegraphics[width=0.48\textwidth]{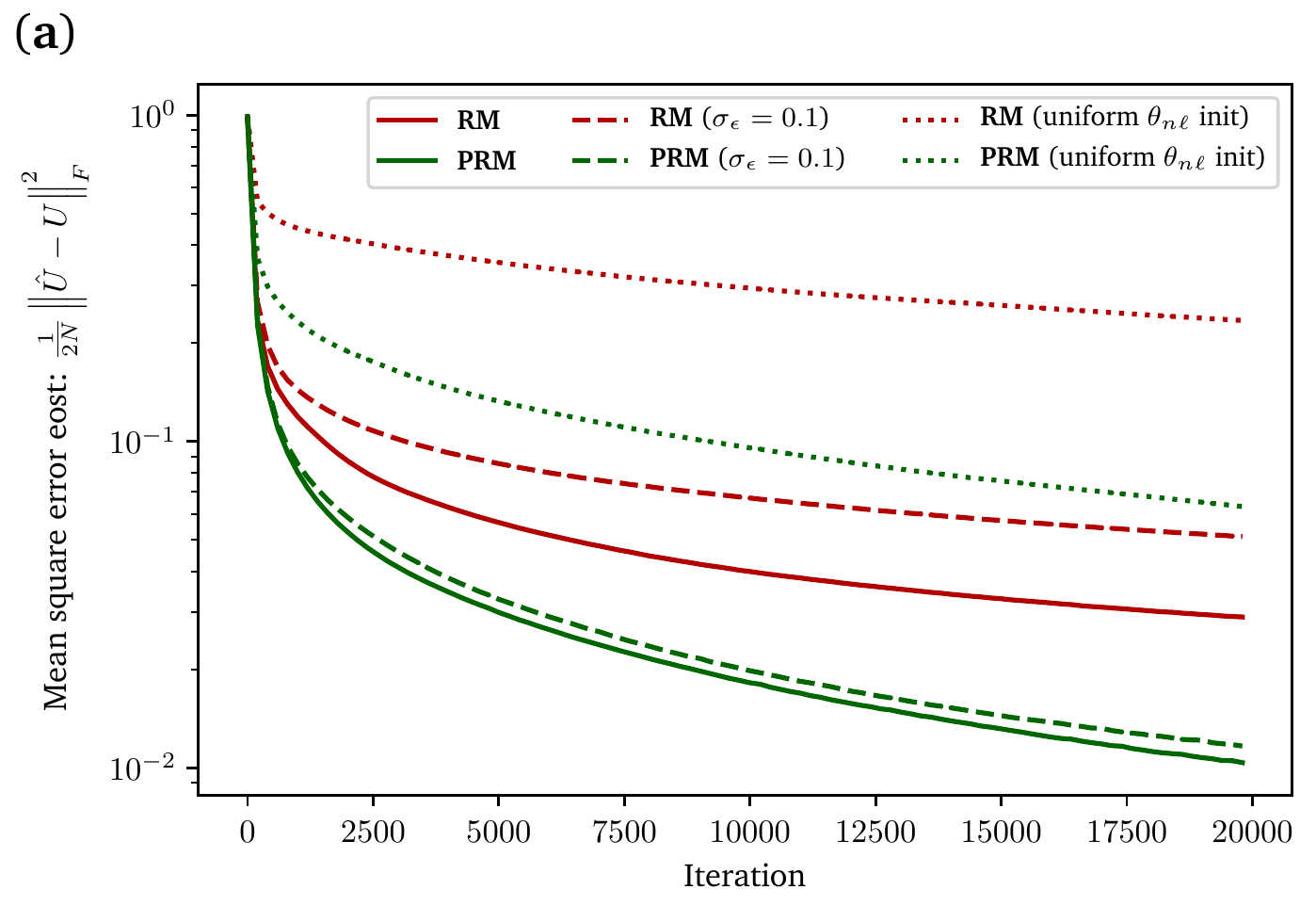}\\
    \includegraphics[width=0.48\textwidth]{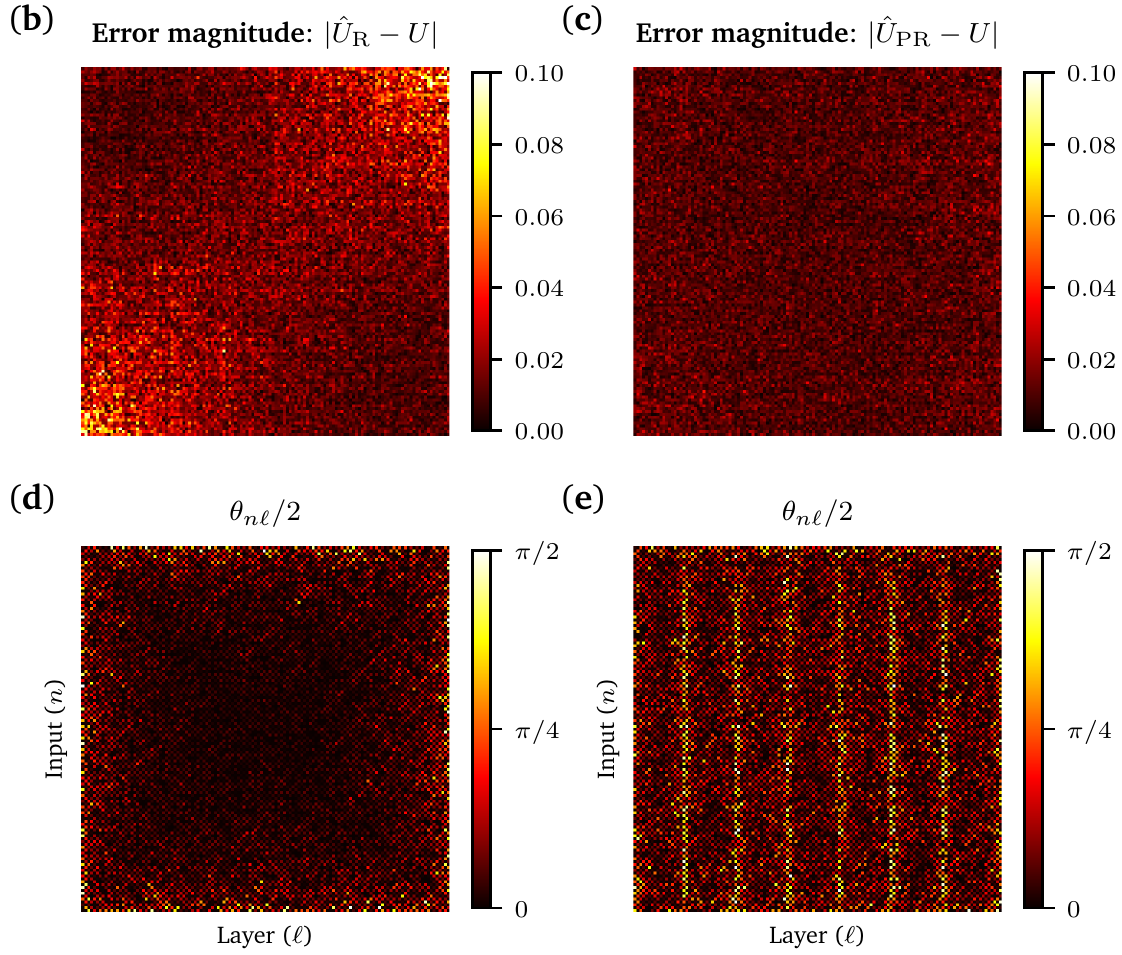}
    \caption{We implement six different optimizations for $N = 128$ where we vary the choice of permuting rectangular mesh (PRM) or rectangular mesh (RM); the initialization (random $\theta_{n\ell}$ or Haar-initialized $\theta_{n\ell}$); and photonic transmissivity error displacements ($\epsilon = 0$ or $\epsilon \sim \mathcal{N}(0, 0.01)$, where $\sigma_\epsilon^2 = 0.01$ is the variance of the beamsplitter errors). Conditions: $20000$ iterations, Adam update, learning rate of $0.0025$, batch size of 256, simulated in \texttt{tensorflow}. (a) Comparison of optimization performance (defaults are Haar initialization and $\epsilon_{n\ell} = 0$ unless otherwise indicated). Optimized error magnitude spatial map for (b) rectangular mesh shows higher off-diagonal errors and than (c) permuting rectangular. The optimized $\theta_{n\ell}$ phase shifts (see Appendix \ref{sec:periodicparameters}) for (d) rectangular meshes are close to zero (cross state) near the center as opposed to (e) permuting rectangular meshes which have a striped pattern (likely due to initialization). NOTE: by $\abs{\cdot}$, we refer to the elementwise norm.}
    \label{fig:unitaryopt}
\end{figure}

\begin{figure}[ht]
    \centering
    \includegraphics[width=0.48\textwidth]{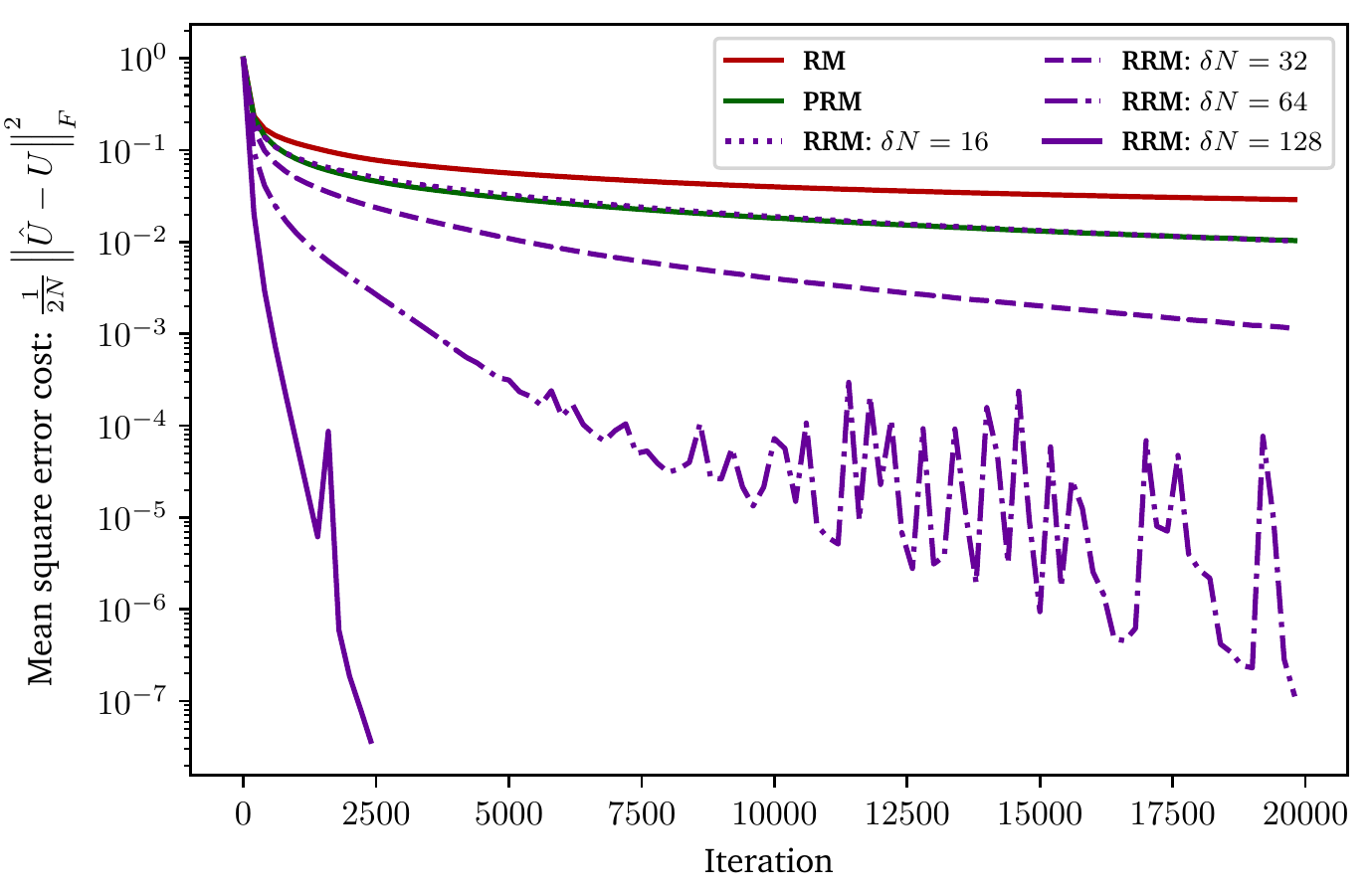}
    \caption{A comparison of test error in \texttt{tensorflow} for $N = 128$ between rectangular (RM), permuting rectangular (PRM), and redundant rectangular (RRM) meshes for: $20000$ iterations, Adam update, learning rate of $0.0025$, batch size of 256. Ideal = Haar random initialized $\theta_{n\ell}$ with $\epsilon = 0$. $\delta N$ is the additional layers added in the redundant mesh. We stopped the $\delta N = 128$ run within 4000 iterations when it reached convergence within machine precision. Redundant meshes with 32 additional layers converge better than permuting rectangular meshes, and with just 16 additional layers, we get almost identical performance.}
    \label{fig:rrm}
\end{figure}

The singular value decomposition (SVD) architecture discussed in Refs. \cite{Miller2013Self-configuringInvited, Shen2017DeepCircuits} consists of optical lossy components flanked on both sides by rectangular meshes and are capable of implementing any linear operation with reasonable device input power. Note that with some modifications (e.g. treating loss and gain elements like nonlinearities in the procedure of Ref. \cite{Hughes2018TrainingMeasurement}), SVD architectures can also be trained physically using \textit{in situ} backpropagation. We simulate the gradient-based optimization of SVD architectures using automatic differentiation in Appendix \ref{sec:svdcomparison}.

\section{Discussion}

\subsection{Haar initialization}
For global optimization and robustness of universal photonic meshes, it is important to consider the required biases and sensitivities for each mesh component. Implementing any Haar random matrix requires that each component independently follows an average reflectivity within some tolerance. This requirement becomes more restrictive with the number of input and output ports accessible by each mesh component. For the rectangular mesh, this means the center mesh components are close to cross state and the most sensitive.

In a Haar-initialized mesh, as shown in Figure \ref{fig:checkerboardplots}, the light injected into a single input port spreads out to all waveguides in the device uniformly regardless of $N$. This is a preferable initialization for global optimization because Haar random matrices require this behavior. In contrast, when randomly initializing phases, the light only spreads out over a limited band of outputs. This band gets relatively small compared to the mesh gets larger as shown in Figure \ref{fig:bandlimits}. 

The average reflectivities given by Haar initialization may be useful for inverse design approaches \cite{Piggott2017Fabrication-constrainedDesign} for compact tunable or passive multiport interferometers. The component tolerances may inform how robust phase shifters need to be given error sources such as thermal crosstalk \cite{Shen2017DeepCircuits}. The thermal crosstalk might make it difficult to achieve required tolerances for devices interfering up to $N = 1000$ modes that generally have phase shift tolerances just above $10^{-2}$ radians.\footnote{Ref. \cite{Shen2017DeepCircuits} propose a standard deviation of $\sim 10^{-3}$ might be possible with further circuit characterization, which might be scalable based on Figure \ref{fig:haarphase}(c).}

In our simulations in Section \ref{sec:simulation}, we assume that the control parameter for photonic meshes is linearly related to the phase shift. However, in many current phase shifter implementations, such as thermal phase shifters \cite{Shen2017DeepCircuits}, the phase is a nonlinear function of the control parameter (i.e., the voltage) and has minimum and maximum values, unlike the unbounded phase used in our optimization. In addition, like the Haar phase in our theory, the voltage acts as the CDF for transmissivities in the physical device, up to a normalization factor. Particular attention needs to be given to phase uncertainty as a function of voltage, since the Haar random distribution of internal MZI phases has small variance for large $N$, as we show in Figure \ref{fig:haarphase}(c). As mentioned in Section \ref{sec:haarinit}, the ideal transmissivity-voltage dependence with this consideration would be identical to the transmissivity vs Haar phase dependence in Figure \ref{fig:haarphase}(d).

\subsection{Applications of mesh optimization}

Meshes can be tuned using either self-configuration \cite{Reck1994ExperimentalOperator, Miller2013Self-configuringInvited} or global optimizations (gradient-based \cite{Hughes2018TrainingMeasurement} or derivative-free \cite{Tang2018ReconfigurableConversion}). The algorithmic optimizations proposed in Refs. \cite{Reck1994ExperimentalOperator, Miller2013Self-configuringInvited} assume that each component in the mesh can cover the entire split ratio range, which is not the case in presence of 50:50 beamsplitter errors. This ultimately leads to lower fidelity in the implemented unitary operation, which can be avoided using a double-MZI architecture \cite{Miller2015PerfectComponents, Wilkes201660dBInterferometer} or a vertical layer-wise progressive algorithm \cite{Miller2017SettingMethod}. We explore a third alternative to overcome photonic errors; gradient-based global optimization is model-free and, unlike algorithmic approaches, can efficiently tune photonic neural networks \cite{Hughes2018TrainingMeasurement}. This model-free property makes gradient-based optimization robust to fabrication error; we show in Figure \ref{fig:unitaryopt}(a) that meshes with split ratio error variances of up to $\sigma_\epsilon = 0.1$ can be optimized nearly as well as a perfect mesh, particularly for permuting rectangular meshes.

In the regime of globally optimized meshes, we propose two strategies to modify the rectangular architecture: adding waveguide permutation layers and adding extra tunable vertical MZI layers. Both approaches relax the cross state requirements on the MZIs and accelerate the mesh optimization process. Nonlocal interference works by allowing inputs that are far away physically in the mesh to interact. These approaches are inspired by several recent proposals in machine learning and coherent photonics to design more error tolerant and efficient meshes, many of which use single layers of MZIs and nonlocal waveguide interactions \cite{Genz1998MethodsMatrices, Mathieu2014FastMatrices, Flamini2017BenchmarkingProcessing, Jing2017TunableRNNs}; such designs can also be considered to be in the same class of permuting architectures as our proposed permuting rectangular mesh. Adding extra tunable vertical layers, as proposed in Ref. \cite{Burgwal2017UsingUnitaries}, simply adds more tunable paths for the light to achieve a desired output. As shown in Figure \ref{fig:unitaryopt}, we achieve up to five orders of magnitude improvement in convergence at the expense of doubling the mesh size and parameter space.

Like permuting rectangular meshes, multi-plane light conversion successfully applies the non-local interference idea for efficient spatial mode multiplexing \cite{Labroille2014EfficientConversion, Labroille2016ModeConversion}. In this protocol, alternating layers of transverse phase profiles and optical Fourier transforms (analogous to what our rectangular permutations accomplish) are applied to reshape input modes of light \cite{Labroille2014EfficientConversion, Labroille2016ModeConversion}. A similar concept is used in unitary spatial mode manipulation, where stochastic optimization of deformable mirror settings allow for efficient mode conversion \cite{Morizur2010ProgrammableManipulation}. Thus, the idea of efficient unitary learning via a Fourier-inspired permuting approach has precedent in contexts outside of photonic MZI meshes.

An on-chip optimization for multi-plane light conversion has been accomplished experimentally in the past using simulated annealing \cite{Tang2018ReconfigurableConversion}. The success of simulated annealing in experimentally training small unitary photonic devices \cite{Tang2018ReconfigurableConversion} (rather than gradient descent as is used in this work) suggests there are other algorithms aside from gradient descent that may effectively enable on-chip training. 

We propose that similar simulated annealing approaches might be made more efficient by sampling Haar phases from uniform distributions and flashing updates onto the device. Similar derivative-free optimizations may also be useful for quantum machine learning \cite{Schuld2019QuantumSpaces, Schuld2018Circuit-centricClassifiers, Killoran2018Continuous-variableNetworks}. Whether such approaches can compete with backpropagation for classical applications remains to be investigated. For experimental on-chip tuning, simulated annealing has the attractive property of only requiring output detectors. For practical machine learning applications, however, there is currently more literature for backpropagation-based optimization. Furthermore, gradient-based approaches allow for continuous control of phase shifters during the optimization.

Our \texttt{tensorflow} simulations may be useful in the design of optical recurrent neural networks (RNNs) that use unitary operators parameterized by photonic meshes. Such ``unitary RNNs" (URNNs) have already been simulated on conventional computers and show some promise in synthetic long-term memory tasks \cite{Jing2017TunableRNNs, Dangovski2017RotationalMemory}. Unitary RNNs are physically implementable using a single mesh with optical nonlinearities and recurrent optoelectronic feedback, suggesting that the architecture discussed in this work is a scalable, energy-efficient option for machine learning applications. It is possible that some tunable features such as the ``bandedness" of unitaries implemented by rectangular MZI meshes can be useful (e.g. as an attention mechanism in sequence data) for certain deep learning tasks that use URNNs.

\section{Conclusion}

The scalability of gradient-based optimization of Haar random unitary matrices on universal photonic meshes is limited by small reflectivities and MZI phase shifter sensitivities arising from the constraint of locally interacting components. As shown in Section \ref{sec:haarinit}, the required average reflectivity and sensitivity for each MZI is inversely related to the total number of inputs and outputs affected by the MZI. If the tolerance requirements are not met by the physical components, optimization algorithms will have difficulty converging to a target unitary operator. As shown in Section \ref{sec:simulation} for the case of $N = 128$, convergence via \textit{in situ} backpropagation is generally not achieved if phase shifters are initialized randomly. However, Haar initialization can sufficiently bias the optimization for convergence to a desired random unitary matrix, even in the presence of significant simulated beamsplitter fabrication errors.

In Section \ref{sec:architectures}, we propose adding extra tunable beamsplitters or mesh nonlocalities to accelerate mesh optimization. Naive (uniform random) initialization on a standard photonic mesh has difficulty learning random unitary matrices via gradient descent. By introducing non-localities in the mesh, we can improve optimization performance without the need for extra parameters. A Haar-initialized redundant architecture can achieve five orders of magnitude less mean square error for a Haar random unitary matrix and decrease optimization time to such a matrix by at least two orders of magnitude, as shown in Figure \ref{fig:rrm}. Our findings suggest that architecture choice and initialization of photonic mesh components may prove important for increasing the scalability and stability of reconfigurable universal photonic devices and their many classical and quantum applications \cite{Annoni2017UnscramblingModes, Shen2017DeepCircuits, Spring2013BosonChip, Killoran2018Continuous-variableNetworks, Schuld2018Circuit-centricClassifiers, Schuld2019QuantumSpaces, Arrazola2019MachineComputers, Miller2013Self-configuringInvited, Miller2013Self-aligningCoupler}.

\section*{Acknowledgements}
This work was funded by the Air Force Office of Scientific Research, specifically for the Center for Energy‐Efficient 3D Neuromorphic Nanocomputing (CEE3N$^2$), Grant No. FA9550-181-1-0186. We would like to thank Tyler Hughes, Momchil Minkov, Nate Abebe, Dylan Black, and Ian Williamson for illuminating discussions.

\bibliography{matrixoptimization}

\appendix

\section{Software}
To reproduce the results of this paper, the reader can be directed to \href{https://github.com/solgaardlab/neurophox}{\texttt{neurophox}}, an open-source Python package that implements the optimizations and simulations of this paper in \texttt{numpy} and \texttt{tensorflow}. The exact code used to generate the results is provided in the \href{https://github.com/solgaardlab/neurophox-notebooks}{\texttt{neurophox-notebooks}} repository.

\section{Derivation of beamsplitter errors} \label{sec:photonicerrorsproof}
Unitary matrices generated by lossless MZIs are prone to errors in beamsplitter fabrication. We introduce the error $\epsilon$ to our expression derived in Equation \ref{eqn:mzidefinition}, which is twice the displacement in beamsplitter split ratio from $50:50$. Beamsplitter gates with error $\epsilon$ are defined as $B_\epsilon = \begin{bmatrix}\rho & i\tau \\ i\tau & \rho\end{bmatrix}$ where $\rho = \sqrt{\frac{1 + \epsilon}{2}}, \tau = \sqrt{\frac{1 - \epsilon}{2}}$ are transmissivity and reflectivity amplitudes that result in slight variations from a $50:50$ beamsplitter. We use this error definition since it is a measurable quantity in the chip; in fact, there are strategies to minimize $\epsilon$ directly \cite{Miller2015PerfectComponents}. The unitary matrix that we implement in presence of beamsplitter errors becomes
\begin{equation}
\begin{aligned}
U_\epsilon &:= R_\phi B_{\epsilon_2} R_\theta B_{\epsilon_1} \\
t_\epsilon &:= |U_{\epsilon, 12}|^2 = |U_{\epsilon, 21}|^2 \\
r_\epsilon &:= |U_{\epsilon, 11}|^2 = |U_{\epsilon, 22}|^2.
\end{aligned}
\end{equation}
If $\epsilon_1 = \epsilon_2 = \epsilon$, which is a reasonable practical assumption for nearby fabricated structures, then solving for $t_\epsilon$ in terms of $t$:
\begin{equation}
\begin{aligned}
t_\epsilon &= 4 |\rho|^2|\tau|^2 t\\
&= 4 t \left(\frac{1}{2} + \frac{\epsilon}{2}\right) \left(\frac{1}{2} - \frac{\epsilon}{2}\right) \\
&= t(1 - \epsilon^2).
\end{aligned}
\end{equation}
Similarly, we can solve for $r_\epsilon$:
\begin{equation}
\begin{aligned}
r_\epsilon &= 1 - t_\epsilon = r + t \cdot \epsilon^2.
\end{aligned}
\end{equation}

As we have discussed in this paper (and as we later show in Figure \ref{fig:resultcomparison}), photonic errors $\epsilon$ (standard deviation of 0.1) can affect the optimized phase shifts for unitary matrices. The above constraints on $r_\epsilon$ and $t_\epsilon$ suggest that limited transmissivity is likely in the presence of fabrication errors, which can inhibit progressive setup of unitary meshes \cite{Miller2015PerfectComponents, Burgwal2017UsingUnitaries}. However, we will later show through \texttt{tensorflow} simulation that \textit{in situ} backpropagation updates can to some extent address this issue using a more sophisticated experimental protocol involving phase conjugation and interferometric measurements \cite{Hughes2018TrainingMeasurement}.

\section{Haar measure} \label{sec:haarmeasure}

\begin{figure}[h]
    \centering
    \includegraphics[width=0.48\textwidth]{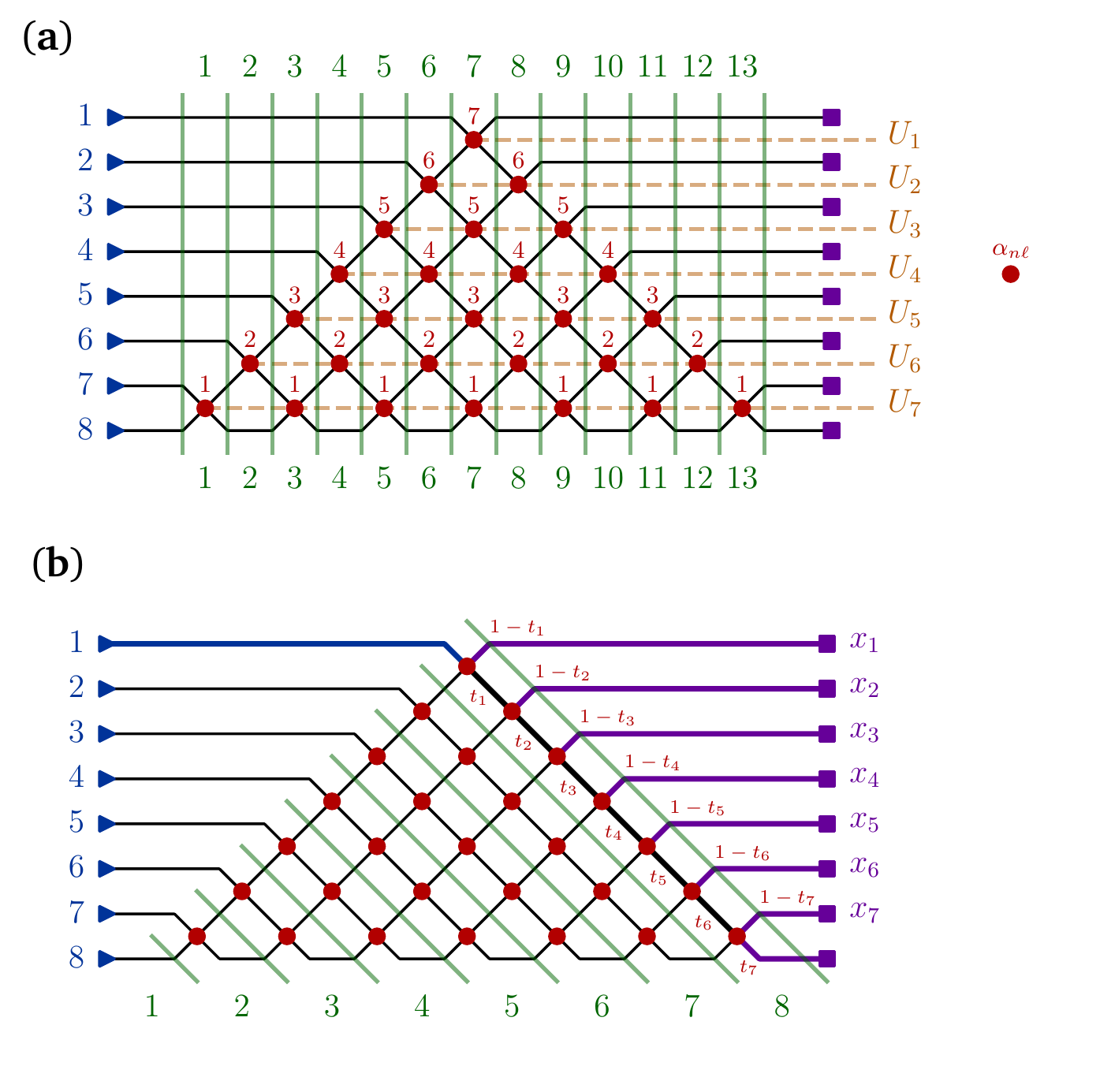}
    \caption{Triangular mesh for $N = 8$ using (a) $2N - 3$ vertical layers $\ell$ showing the sensitivity index $\alpha_{n\ell}$ and (b) $N$ diagonal layers $m$ showing the transmissivity basis ($t_n$ in red) and the measurement basis ($x_n$ in purple).}
    \label{fig:triangularmesh}
\end{figure}

In this section, we outline a proof for the Haar measure of a unitary matrix in terms of the physical parameters of a photonic mesh to supplement our discussion of Haar phase and the proof in Ref. \cite{Russell2017DirectMatrices}. The Haar measure for $\mathrm{U}(N)$ can be defined in two physical basis representations: the measurement basis represents measurements after each MZI and the transmissivity basis represents the transmissivity of each MZI.

To make our explanation simpler, we will adopt the orthogonalization protocol used by Ref. \cite{Reck1994ExperimentalOperator}. In this representation, we define the triangular mesh $U_\mathrm{T}$ as
\begin{equation}\label{eqn:reck}
    \begin{aligned}
        U_{\mathrm{T}} &= \prod_{m = 0}^{N - 1} U^{(N - m)} \\
        U^{(m)} &= \prod_{n=1}^{m-1} U_{N - n}(\theta^{(m)}_{N - n}, \phi^{(m)}_{N - n}) \cdot D_{m}(\gamma_{N - m + 1}),
    \end{aligned}
\end{equation}
where $D_{m}$ is a diagonal matrix representing a single mode phase shift at index $N - m + 1$.

The $N$ operators $U^{(m)}$ represent the diagonal layers of the triangular mesh and their role is to project inputs from Hilbert space dimension from $m$ to $m - 1$ recursively until we reach a single mode phase shift in $U^{(1)} = D_1(\gamma_N)$. Our proof moves the same direction as Reck's orthogonalization procedure; starting from $m = N$, we progressively find the for each $U^{(m)}$ in decreasing order. For each layer $m$, there are $2m - 1$ complex hyperspherical coordinates ($m - 1$ ``amplitude" coordinates and $m$ ``phase" coordinates). The first column vector of $U$ can be recovered by shining light (using a unit power $P = 1$) through the top port of the layer (given by $n = N - m + 1$) and measuring the output fields in the triangular mesh generated by $U^{(m)}$, as shown in Figure \ref{fig:triangularmesh}(b). As mentioned in Refs. \cite{Reck1994ExperimentalOperator, Miller2013Self-configuringInvited}, progressive optimization moves in the opposite direction; the desired output fields are shined back into the device and the transmissivities $t^{(m)}_n$ and phases $\phi^{(m)}_n$ for each layer $m$ (moving from $N$ to $1$) can be progressively tuned until all the power lies in the top input port for that layer.

The measurement basis is an unbiased Haar measure (as shown in Ref. \cite{Russell2017DirectMatrices} using Gaussian random vectors) and can be physically represented by the power $x_n$ measured at waveguides $n \leq m - 1$ due to shining light through the top input port for that layer. Unlike the proof in Ref. \cite{Russell2017DirectMatrices}, we choose our constraint that the input power $P = 1$ rather than $P \in \mathbb{R}^+$, which introduces a normalization prefactor in our Haar measure by integration over all possible $P$.\footnote{This prefactor is exactly $\int_{0}^{\infty} dP e^{-P}P^{m - 1} = (m - 1)!$.} This allows us to ignore the power in the final output port $x_N$ because energy conservation ensures we have the constraint $x_N = 1 - \sum_{n=1}^{N - 1} x_n$. Therefore, our simplified Cartesian basis for each $m$ is (ignoring the normalization prefactor):
\begin{equation}
    \begin{aligned}
    \mathrm{d}\mu(U^{(m)}) &\propto \mathrm{d}\gamma_{N - m} \prod_{n=1}^{m-1} \mathrm{d}x_n \prod_{n=1}^{m}\mathrm{d}\phi_{n}.
    \end{aligned}
\end{equation}

Now we represent the Cartesian power quantities $x_n$ explicitly in terms of the component transmissivities, which we have defined already to be $t_{n} := \cos^2(\theta_{n} / 2)$. Using the same convention as hyperspherical coordinates, we get the following recursive relation for $x_n$ as shown diagrammatically by following the path of light from the top input port in Figure \ref{fig:triangularmesh}(b):
\begin{equation} \label{eqn:carttocomp}
    \begin{aligned}
    x_n &= (1 - t_n) \prod_{k =1}^{n-1} t_{k}.
    \end{aligned}
\end{equation}

Intuitively, Equation \ref{eqn:carttocomp} implies that the power $x_n$ measured at port $n$ is given by light that is transmitted by the first $n - 1$ components along the path of light and then reflected by the $n$th component. In other words, $x_n$ follows a geometric distribution.

We can use Equation \ref{eqn:carttocomp} to find the Jacobian $\mathcal{J} \in \mathbb{R}^{N - 1 \times N - 1}$ relating the $x_n$ and the $t_n$. We find that we have a lower triangular matrix $\mathcal{J}$ with diagonal elements for $n \leq N - 1$
\begin{equation} \label{eqn:jacobian}
    \begin{aligned}
    \mathcal{J}_{nn} &= \frac{\partial x_n}{\partial t_n} = -\prod_{k =1}^{n-1} t_{k}.
    \end{aligned}
\end{equation}
We know $\mathcal{J}$ is lower triangular since for all $n' > n$, $\mathcal{J}_{nn'} = \frac{\partial x_n}{\partial t_{n'}} = 0$ from Equation \ref{eqn:carttocomp}.

Since the determinant of a lower triangular matrix is the same as the product of the diagonal, we can directly evaluate the unbiased measure (off by a normalization constant) as
\begin{equation} \label{eqn:haarmeasurecomp}
    \begin{aligned}
    \mathrm{d}\mu(U^{(m)}) &\propto \mathrm{d}\gamma_{N - m + 1} \det \mathcal{J} \prod_{n=1}^{m - 1} \mathrm{d}t_n \prod_{n=1}^{m} \mathrm{d}\phi_{n} \\
    &= \mathrm{d}\gamma_{N - m + 1} \prod_{n = 1}^{m-1} \mathcal{J}_{nn} \prod_{n=1}^{m-1} \mathrm{d}t_n \prod_{n=1}^{m} \mathrm{d}\phi_{n} \\
    &\propto \mathrm{d}\gamma_{N - m + 1} \prod_{n = 2}^{m - 1} t_{n - 1}^{m - n} \prod_{n=1}^{m-1} \mathrm{d}t_n \prod_{n=1}^{m} \mathrm{d}\phi_{n}
    \end{aligned}
\end{equation}

To get the total Haar measure, we multiply the volume elements for the orthogonal components $\mathrm{d}\mu(U^{(m)})$. We get from this procedure that the sensitivity index $\alpha_{n\ell} = N - n$ for a triangular mesh in Equation \ref{eqn:haarmeasurecomp} (independent of $\ell$), which can be seen using Figure \ref{fig:triangularmesh}. We can express this Haar measure in terms of $\mathcal{Q}_{\alpha_{n\ell}}(t_{n\ell})$, the probability distribution for the transmissivity, and $\mathcal{P}_{\alpha_{n\ell}}(\theta_{n\ell} / 2)$, the probability distribution for the phase shift corresponding to that same transmissivity, assuming appropriate choice $n, \ell$ for the triangular mesh:
\begin{equation} \label{eqn:haarmeasurefull}
    \begin{aligned}
    \mathrm{d}\mu(U) &= \prod_{n=1}^{N} \mathrm{d}\mu(U^{(n)}) \\
    &= \prod_{n} \mathrm{d}\gamma_n \prod_{n, \ell} \mathcal{Q}_{\alpha_{n\ell}}\left(t_{n\ell}\right) dt_{n\ell} \mathrm{d}\phi_{n\ell} \\
    &= \prod_{n} \mathrm{d}\gamma_n \prod_{n, \ell} \mathcal{P}_{\alpha_{n\ell}}\left(\frac{\theta_{n\ell}}{2}\right) \mathrm{d}\theta_{n\ell} \mathrm{d}\phi_{n\ell}
    \end{aligned}
\end{equation}

We can now normalize Equation \ref{eqn:haarmeasurecomp} using the normalization factor for $P$ to get $\mathcal{Q}_{\alpha_{n\ell}}(t_{n\ell})$ and then substitute $t_{n\ell} = \cos^2(\theta_{n\ell} / 2)$ to get our desired expression for $\mathcal{P}_{\alpha_{n\ell}}(\theta_{n\ell} / 2)$:

\begin{equation}\label{eqn:haarpdfs}
    \begin{aligned}
    \mathcal{Q}_{\alpha_{n\ell}}\left(t_{n\ell}\right) &= \alpha_{n\ell} t_{n\ell}^{\alpha_{n\ell} - 1} \\
    \mathcal{P}_{\alpha_{n\ell}}\left(\frac{\theta_{n\ell}}{2}\right) &= \alpha_{n\ell}\sin\left(\frac{\theta_{n\ell}}{2}\right) \left[\cos\left(\frac{\theta_{n\ell}}{2}\right)\right]^{2\alpha_{n\ell} - 1}.
    \end{aligned}
\end{equation}

Finally, we can recover the Haar phase parameter $\xi_{n\ell} \in [0, 1]$ (i.e. the cumulative density function) in terms of either $t_{n\ell}$ or $\theta_{n\ell}$: 

\begin{equation}\label{eqn:physicalhaarphase}
    \begin{aligned}
    \xi_{n\ell} &= \left[\cos\left(\frac{\theta_{n\ell}}{2}\right)\right]^{2\alpha_{n\ell}} = t_{n\ell}^{\alpha_{n\ell}}.
    \end{aligned}
\end{equation}

Finally, as explained in Ref. \cite{Russell2017DirectMatrices}, we can use the Clements decomposition \cite{Clements2016AnInterferometers} to find another labelling for $\alpha_{n\ell}$ in a rectangular mesh that gives probability distributions and Haar phases in the same form as Equations \ref{eqn:haarpdfs} and \ref{eqn:physicalhaarphase} respectively.

\section{Unitary bandsizes} \label{sec:bandsize}

We would like to quantify the bandedness of matrices implemented by the meshes with randomly initialized phases. We define the $\eta$-bandsize as the minimum number of matrix elements whose absolute value squared sums to $(1 - \eta)N$. Note that our $\eta$-bandsize measurement is agnostic of the ordering of the inputs and outputs, and is therefore agnostic to any permutations that may be applied at the end of the decomposition. In photonics terms, if $\eta = 0.001$, let $r_i$ measure the fraction of output waveguides over which $99.9\%$ of the power is distributed when light is input into waveguide $i$. The $\eta$-bandsize is $r_i$ averaged over all $i$. Sampling from our matrix distributions, we observe the relationship between the bandsize (given $\eta = 0.001$) and the dimension $N$ in Figure \ref{fig:bandlimits}.

\begin{figure}[h]
    \centering
    \includegraphics[width=0.46\textwidth]{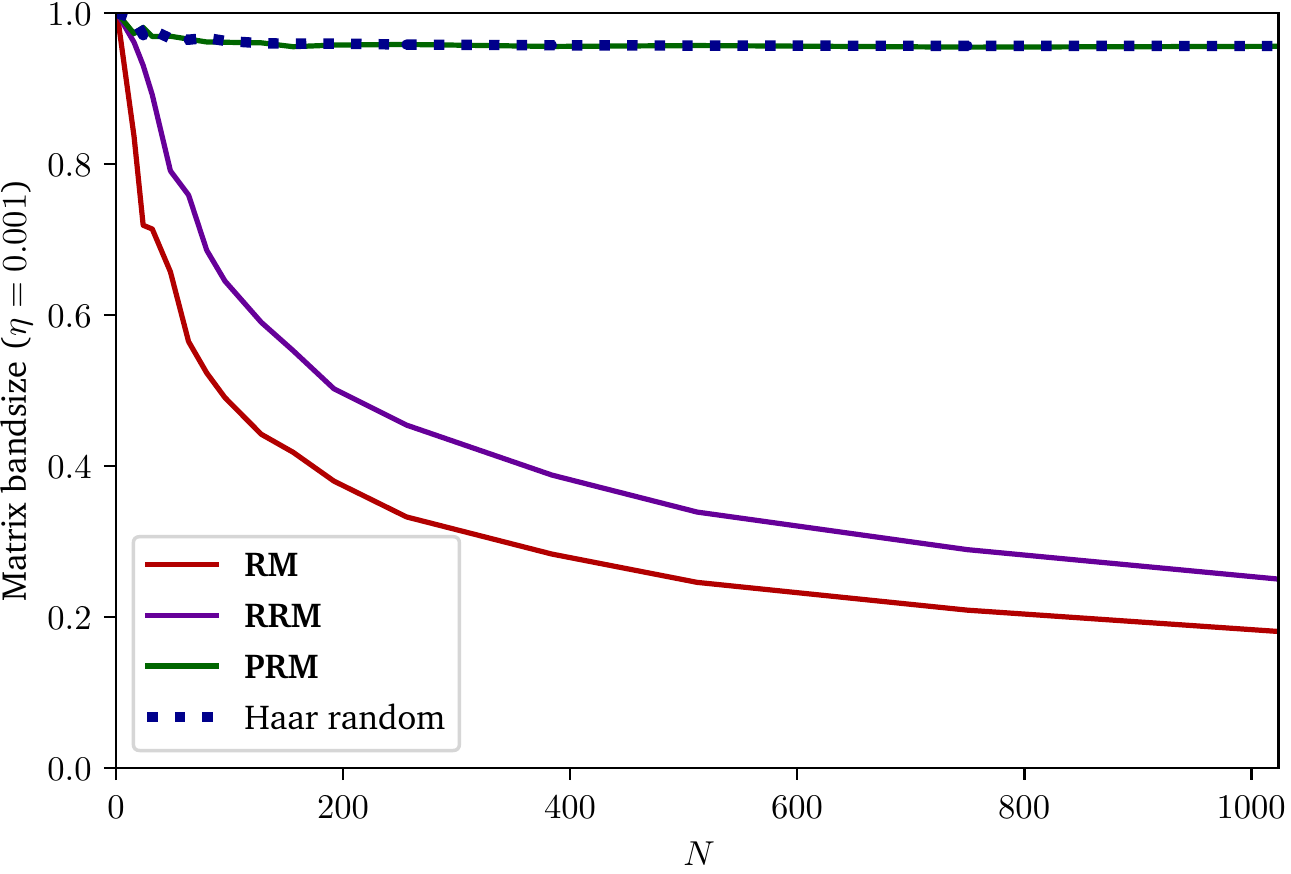}
    \caption{Given $\eta = 0.001$, we compare bandsizes for rectangular ($U \sim \mathcal{U}_\mathrm{R}(N, N)$), permuting rectangular ($U \sim \mathcal{U}_\mathrm{PR}(N)$), and redundant meshes ($U \sim \mathcal{U}_\mathrm{R}(N, 2N)$). Permuting rectangular meshes match the bandsize of Haar random matrices.}
    \label{fig:bandlimits}
\end{figure}

\section{Introducing photonic errors in a redundant mesh} \label{sec:rrmphotonicerrors}

\begin{figure}[ht]
    \centering
    \includegraphics[width=0.48\textwidth]{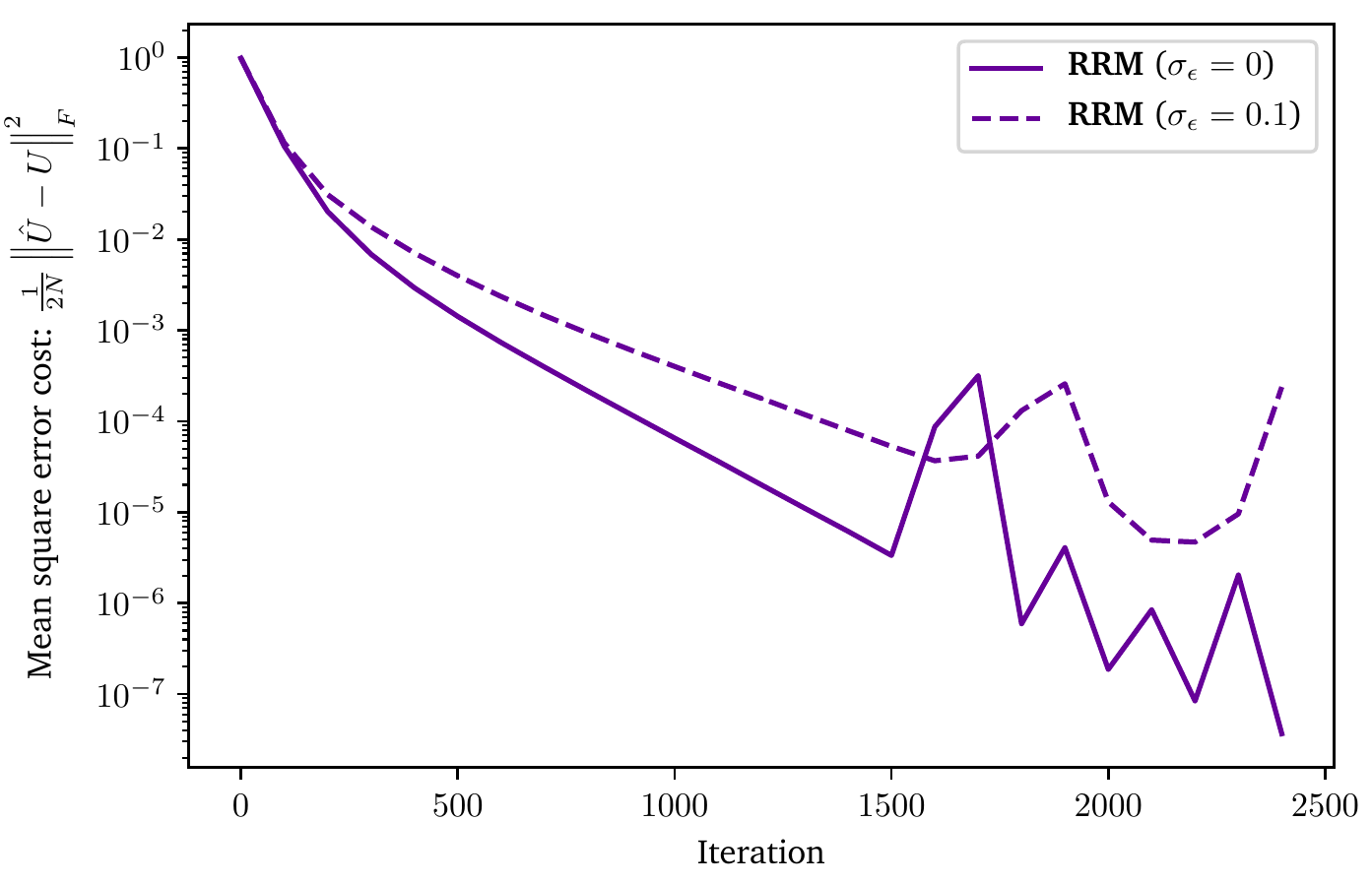}
    \caption{A comparison of test mean square error for $N = 128$ between redundant rectangular meshes with error $\epsilon$ for $256$-layer mesh for: $20000$ iterations, Adam update, learning rate of $0.0025$, batch size of 256, simulated in \texttt{tensorflow}.}
    \label{fig:rrmphotonic}
\end{figure}

When photonic errors are added to the redundant mesh, specifically the $256$-layer mesh, we observe a slight decrease in optimization performance in Figure \ref{fig:rrmphotonic}, similar to what we observed for the rectangular and permuting rectangular meshes in Figure \ref{fig:rrm}. This decrease in performance, however, is less concerning considering that we still achieve a mean square error of around $10^{-5}$, suggesting that RRM might be more robust to photonic errors even during on-chip optimization.

\begin{figure}[ht]
    \centering
    \includegraphics[width=0.48\textwidth]{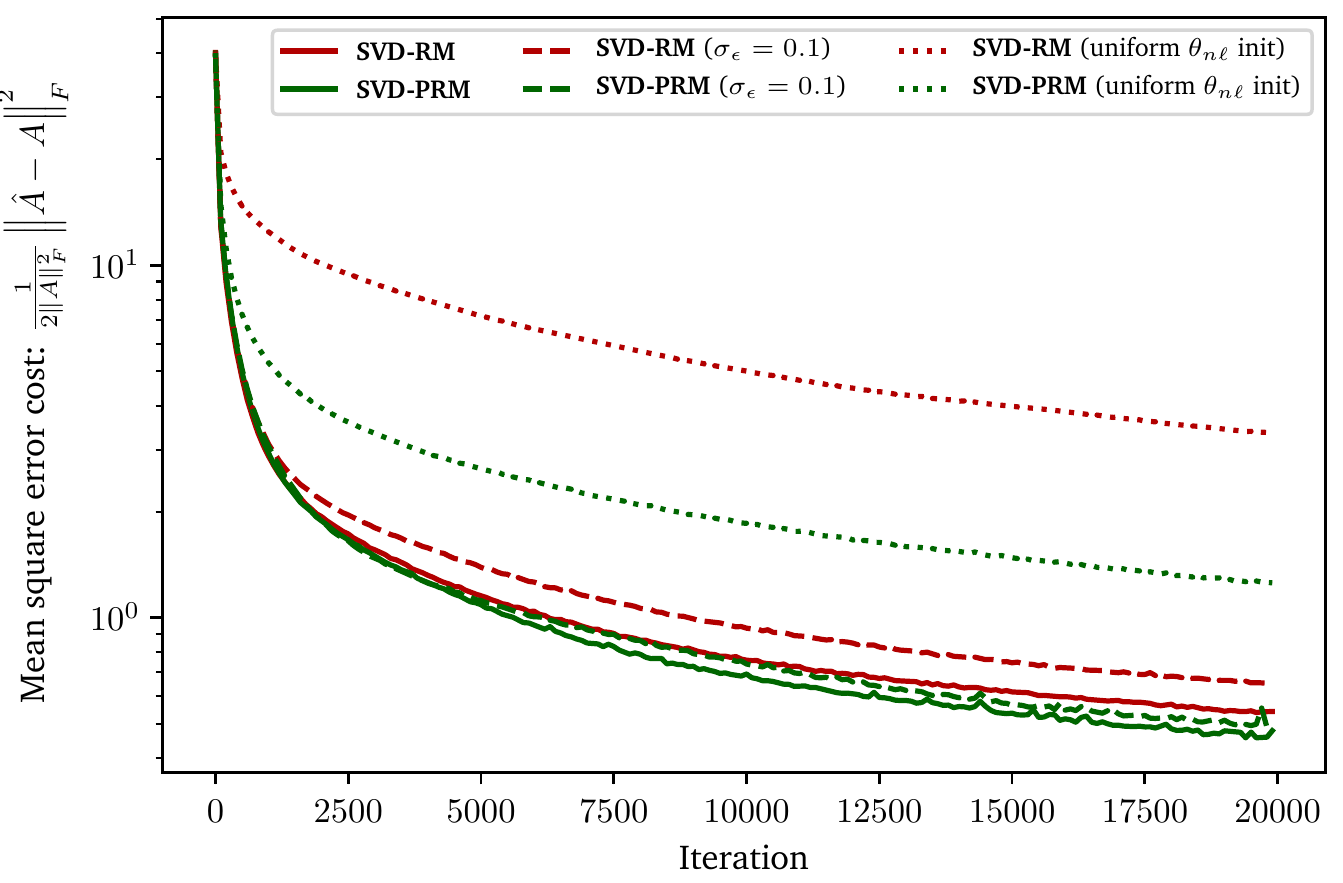}
    \caption{A comparison of test mean square error for $N = 64$ between SVD devices using rectangular (SVD-RM) and permuting rectangular (SVD-PRM) meshes for: $20000$ iterations, Adam update, learning rate of $0.005$, batch size of $128$, simulated in \texttt{tensorflow}. Unless otherwise noted, the default setting is Haar random initialized $\theta_{n\ell}$ with $\sigma_\epsilon = 0$.}
    \label{fig:svd}
\end{figure}

\section{Photonic singular value decomposition simulations} \label{sec:svdcomparison}
We compare the simulated performance of such rectangular and permuting rectangular architectures in the singular value decomposition (SVD) configuration discussed in Refs. \cite{Miller2013Self-configuringInvited, Shen2017DeepCircuits}. Such architectures would allow one to perform arbitary linear operations with a relatively small footprint, and may have some other useful dimensionality-reduction properties in machine learning contexts.

In SVD, we represent complex matrix $\hat A \in \mathbb{C}^M \times \mathbb{C}^N$ as $\hat A = \hat U \hat \Sigma \hat V^\dagger$, where $\hat \Sigma$ is a diagonal matrix implemented on-chip with $\min(M, N)$ single-mode gain or attenuating elements and $\hat U, \hat V^\dagger$ are unitary matrices implemented in a photonic mesh. While $\hat A$ has $2MN$ free parameters, any global optimization for a photonic SVD implementation using rectangular meshes can have at most $D = N(N - 1) + M(M - 1) + 2\min(N, M) \geq 2MN$ free parameters, with equality when $M = N$. In the triangular architecture discussed in Ref. \cite{Miller2013Self-configuringInvited}, the total complexity of parameters can be exactly $D = 2MN$ when setting a subset of the beamsplitters to bar state. In the case where the total number of singular values for $\hat{A}$ is $S < \min(M, N)$, we get $D = 2S(M + N - S)$ tunable elements. Additionally, there is an ``effective redundancy" in that some vectors in $U, V$ are more important than others due to the singular values.

In our simulations, we investigate an SVD architecture for $A = U \Sigma V^\dagger$ for $A \in \mathbb{C}^M \times \mathbb{C}^N$ composed of the unitaries $U \in \mathbb{C}^M \times \mathbb{C}^M$ and $V \in \mathbb{C}^N \times \mathbb{C}^N$. Note that such an architecture is redundant when $M \neq N$, so we focus on the simple case of $M = N = 64$.

We define our train and test cost functions analogous to the unitary mean-squared error cost functions as
\begin{equation}
\begin{aligned}
    \mathcal{L}_\mathrm{test} &= \frac{N\| \hat A - A\|_F^2}{2\|A\|_F^2}\\
    \mathcal{L}_\mathrm{train} &= \|\hat A X - A X\|_F^2,
\end{aligned}
\end{equation}
where $\hat A = \hat U \hat \Sigma \hat V^\dagger$ is defined in Section \ref{sec:simulation}.

We randomly generate $A \in \mathbb{C}^N \times \mathbb{C}^M$ by expressing $A_{jk} = a + i b$, where $a, b \sim \mathcal{N}(0, 1)$. The synthetic training batches of unit-norm complex vectors are represented by $X \in \mathbb{C}^{N \times 2N}$. 

Assuming a procedure similar to \cite{Hughes2018TrainingMeasurement} can be used in presence of gains and optimization, the permuting rectangular mesh converges slightly faster but is significantly more resilient to uniform random phase initialization compared to the rectangular mesh as shown in Figure \ref{fig:svd}. Both optimizations are minimally affected by beamsplitter error, unlike what is seen in the unitary optimization case.

\section{Periodic parameters} \label{sec:periodicparameters}
We comment on our reported values of $\theta_{n\ell}$ in the checkerboard plots in Figures \ref{fig:haarphase} (of the main text) and \ref{fig:resultcomparison}. Since our simulated optimization does not have the explicit constraint that $\theta_{n\ell} \in [0, \pi)$, we report the ``absolute $\theta_{n\ell}$," where we map all values of $\theta_{n\ell}/2$ to some value in $[0, \pi / 2]$. This corresponds to the transformation (assuming $\theta_{n\ell}$ is originally between $0$ and $2\pi$):
\begin{equation}
    \begin{aligned} \label{eqn:absolutetheta}
        \theta_{n\ell} \to \begin{cases}
        \theta_{n\ell}  & \theta_{n\ell} \leq \pi\\
        2\pi - \theta_{n\ell} & \theta_{n\ell} > \pi\\
        \end{cases}.
    \end{aligned}
\end{equation}

Note a similar treatment as Equation \ref{eqn:absolutetheta} can be used to represent the Haar phase $\xi \in [0, 1]$ in terms of a ``periodic" Haar phase $\widetilde{\xi} \in [0, 2]$ with period 2:
\begin{equation} \label{eqn:periodichaar}
    \begin{aligned}
        \xi(\widetilde{\xi}) &= \begin{cases}
        \widetilde{\xi}  & \widetilde{\xi} \leq 1\\
        2 - \widetilde{\xi} & \widetilde{\xi} > 1
        \end{cases}.
    \end{aligned}
\end{equation}

Note both $\widetilde{\xi}$ and $\widetilde{\theta}$ can therefore be made to vary continuously from $(-\infty, \infty)$ with $\widetilde{\xi}$ having a period of 2 and $\widetilde{\theta}$ having a period of $2\pi$. We map these periodic parameters to their half-periods according to Equations \ref{eqn:absolutetheta} and \ref{eqn:periodichaar} based on symmetry arguments.

\section{Training simulation comparisons} \label{sec:unitaryoptcomparison}

\begin{figure*}[t]
    \centering
    \includegraphics[width=0.4\textwidth]{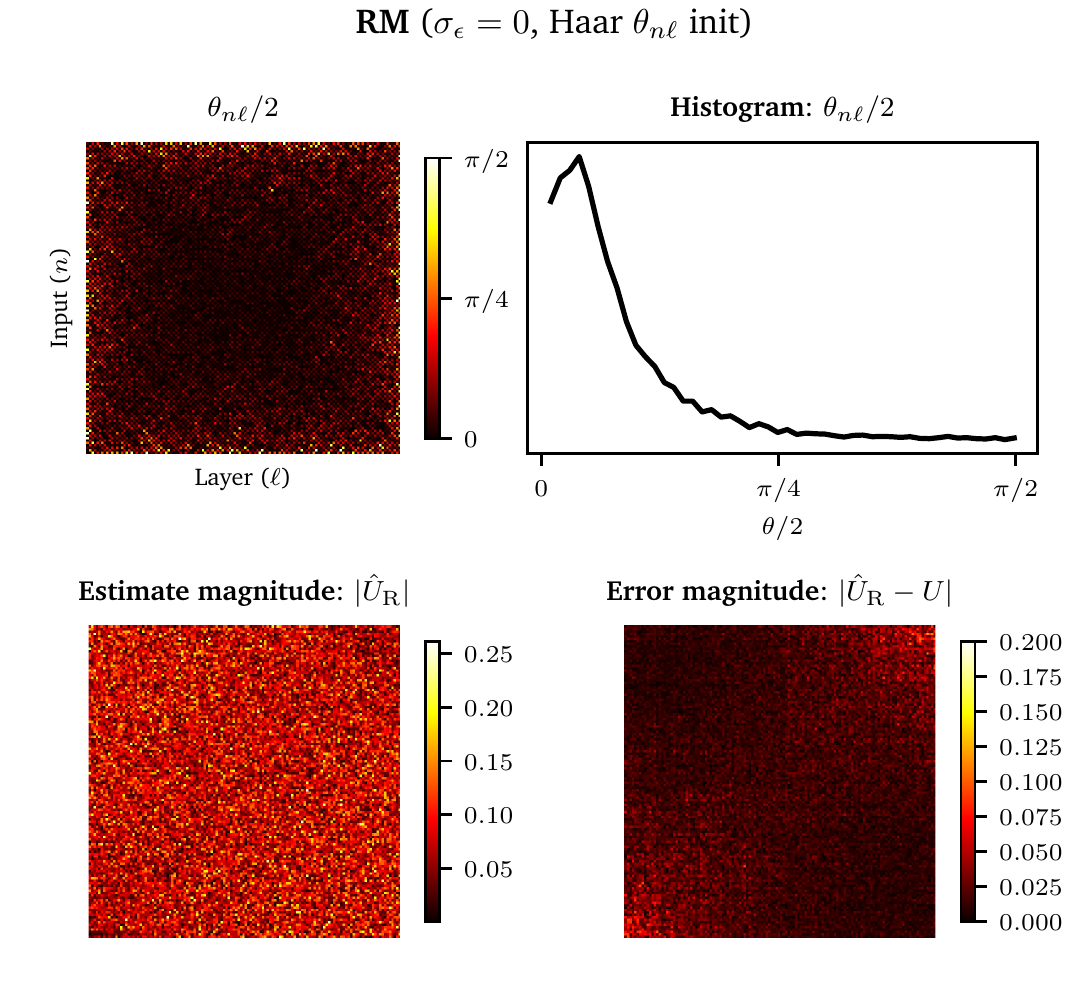}
    \includegraphics[width=0.4\textwidth]{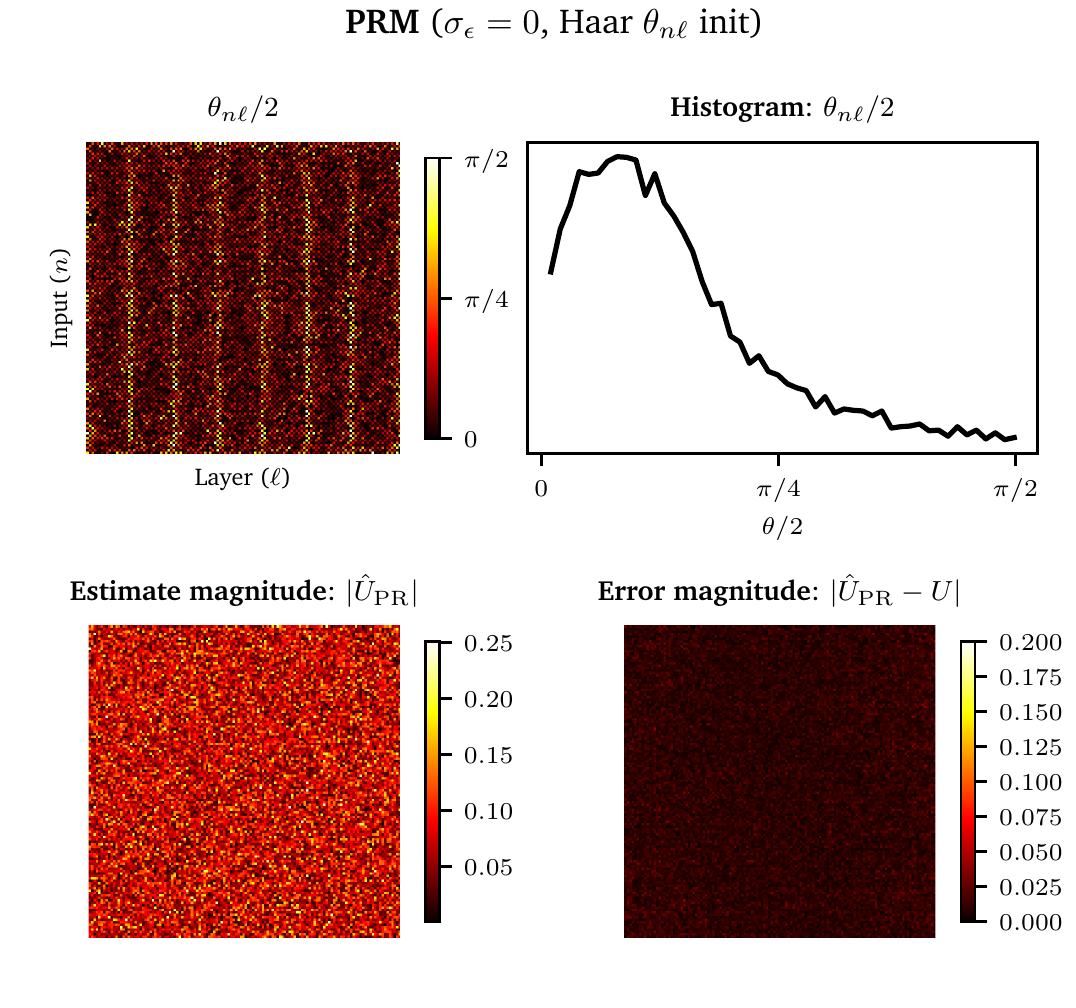}
    \includegraphics[width=0.4\textwidth]{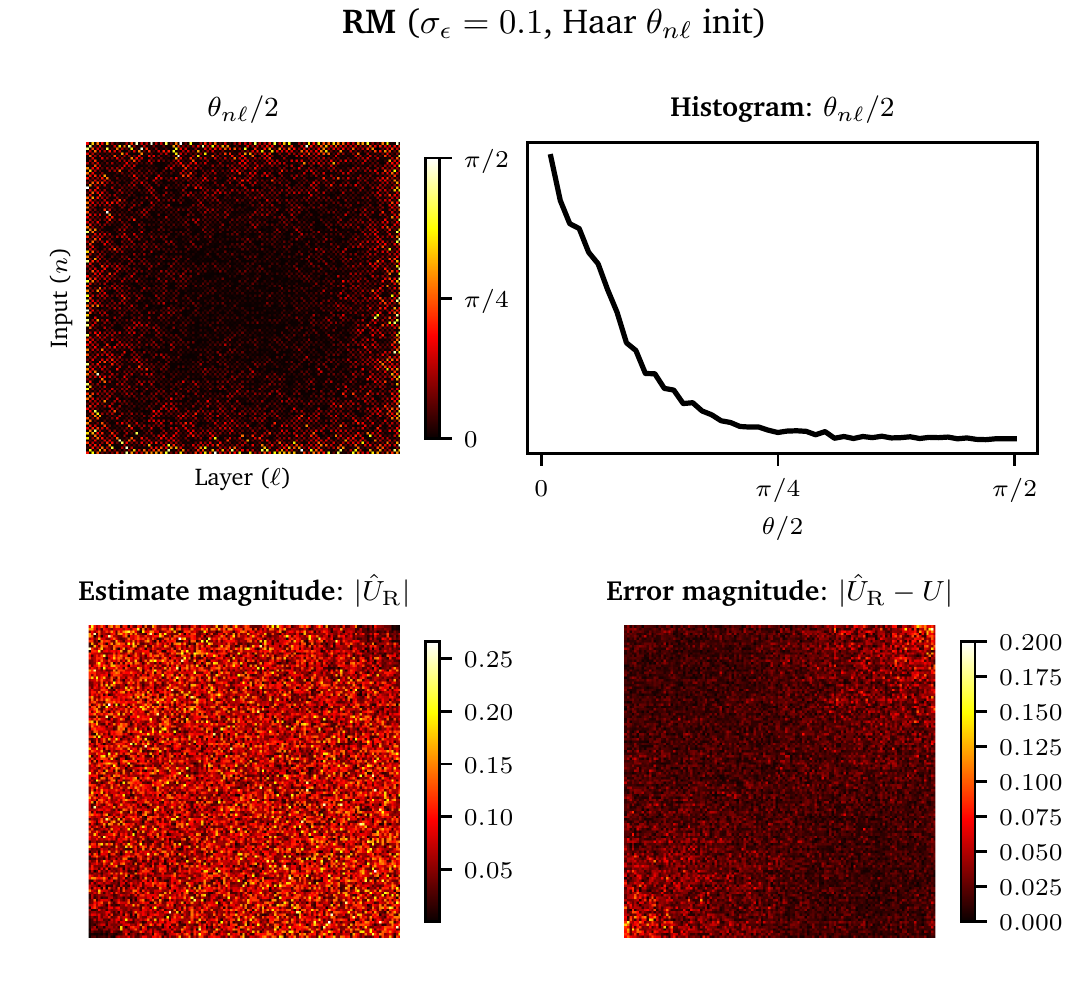}
    \includegraphics[width=0.4\textwidth]{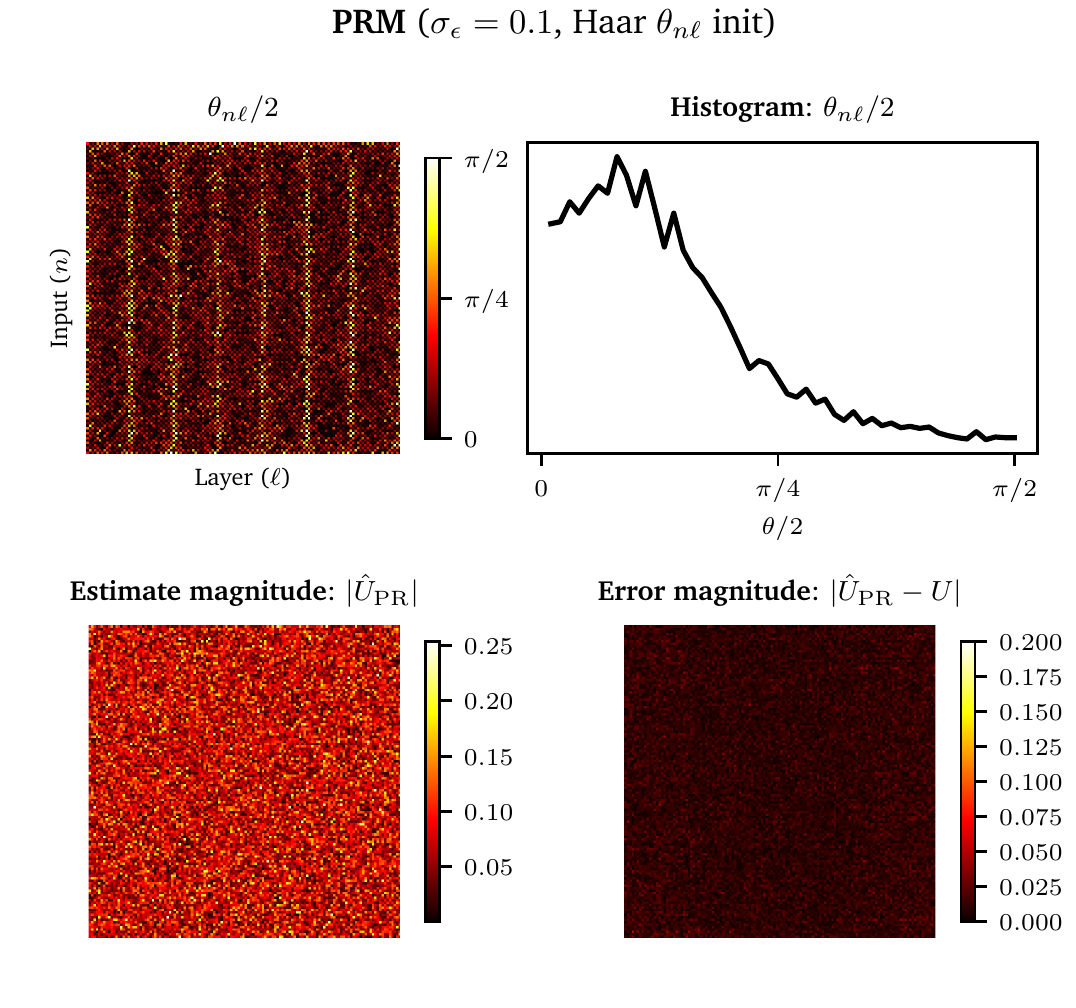}
    \includegraphics[width=0.4\textwidth]{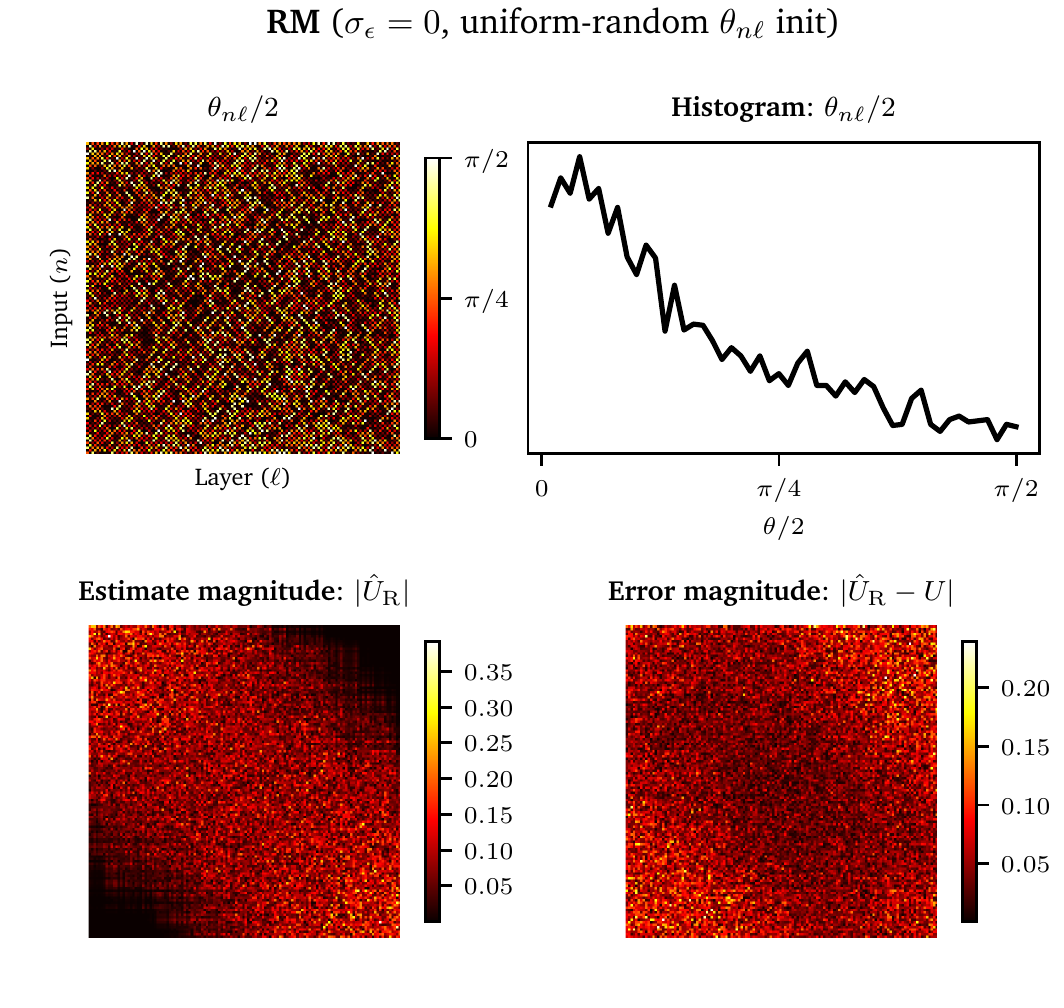}
    \includegraphics[width=0.4\textwidth]{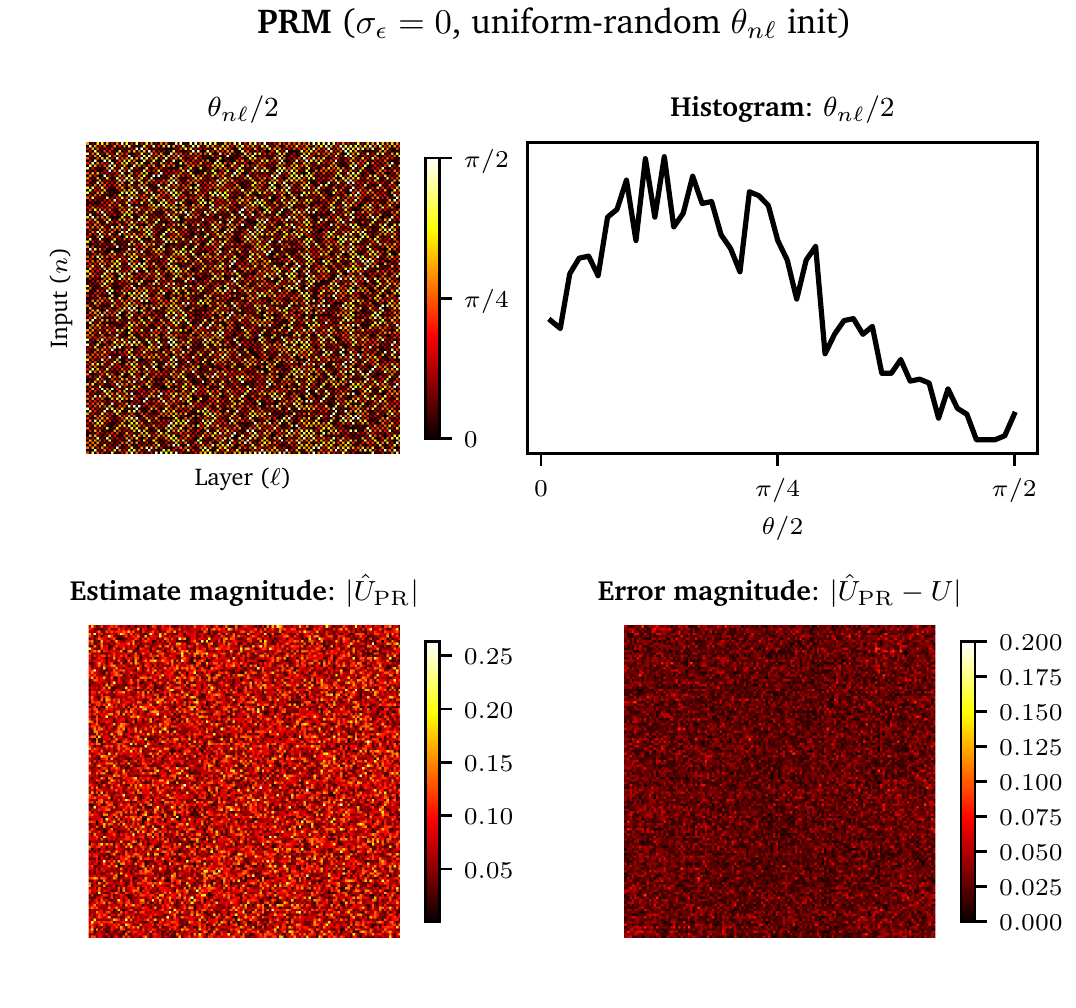}
    \caption{Comparison of learned matrix errors and learned $\theta_{n\ell}$ weights after $20000$ iterations for the Adam update at learning rate $0.0025$ and batch size $256$ for the simple unitary network. We consider two meshes: (1) rectangular mesh (RM), and (2) permuting rectangular mesh (PRM). We consider three conditions for each mesh: (1) ideal (with Haar random unitary initialization), (2) photonic beamsplitter error displacement $\epsilon \sim \mathcal{N}(0, 0.01)$, (3) random initialization.}
    \label{fig:resultcomparison}
\end{figure*}

In Figure \ref{fig:resultcomparison}, we compare the performance for our unitary network experiment over our aforementioned conditions in Section \ref{sec:simulation}. For each plot, we also have an associated video, showing how the parameter distributions, estimates, and errors vary during the course of the optimization, available online.\footnote{See \url{https://av.tib.eu/series/520/photonic+optimization}.} 

There are several takeaways from these plots. First, the reflectivity of the MZIs near the center of the mesh are much smaller in the optimized rectangular meshes than in the permuting rectangular meshes. Second, the gradient descent algorithm has a hard time finding the regime of Haar random matrices after a uniform random phase initialization. The values of $\theta_{n\ell}$ are much larger than they need to be even 100 iterations into the optimization. This is likely evidence of a ``vanishing gradient" problem when the mesh is not Haar-initialized. Finally, an important observation for the meshes with beamsplitter error is that the $\theta_{n\ell} / 2$ distribution shifts slightly towards $0$ in the rectangular mesh. This is a consequence of the limits in reflectivity and transmissivity in each MZI due to beamsplitter fabrication error as discussed in Section \ref{sec:meshintro}.

Our simulated permuting rectangular implementation uses the same layer definitions as defined in Equation \ref{eqn:permutingrectangularmesh} except the $P_k$ with the most layers are in the center of the mesh, and the $P_k$ with the fewest layers are near the inputs and outputs of the mesh. In Figure \ref{fig:architectures}, $P_2$ and $P_3$ would be switched, and for $N = 128$, the order is $[P_2, P_4, P_6, P_5, P_3, P_1]$. We find this configuration to be the best permuting rectangular mesh so far in our experiments, although the architecture in Equation \ref{eqn:permutingrectangularmesh} gives improvements over the rectangular mesh.

\begin{figure}[ht]
    \centering
    \includegraphics[width=0.48\textwidth]{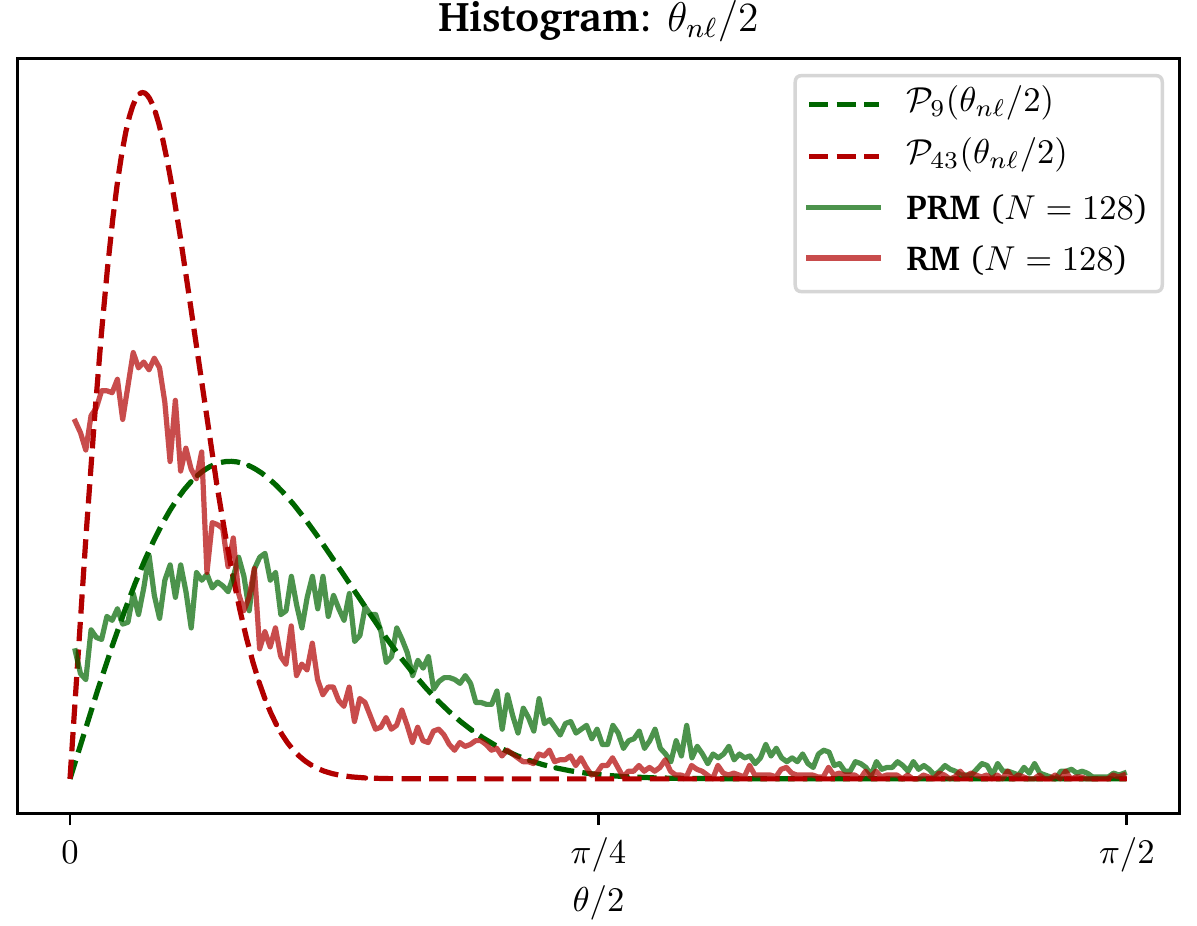}
    \caption{Comparison of learned, normalized $\theta_{n\ell}$ distributions for $N = 128$ rectangular (RM) and permuting rectangular (PRM) meshes with $\mathcal{P}_\alpha(\theta / 2)$ PDFs for $\alpha = \frac{N + 1}{3} = 43$ (the average sensitivity index) and $\alpha = \floor{\frac{N}{2\log N}} = 9$ respectively. Note that permuting meshes have a larger tolerance, which eventually results in faster mesh optimization.}
    \label{fig:rmprmtheta}
\end{figure}

\section{An equivalent definition for \texorpdfstring{$\alpha_{n\ell}$}{anl}} \label{sec:haarproof}

Let $\alpha_{n\ell}$ be the sensitivity index for an MZI (``node'') at (waveguide, layer) coordinates $\left(n, \ell\right)$ in a local decomposition for an $N\times N$ unitary operator. We define the ``row coordinate'' or waveguide index $n$ from the MZI's operator $U_n$ coupling waveguides $n$ and $n+1$, and we define the ``column coordinate'' or layer index $m$ to be $\ell=k+1$, where $k$ is the maximum number of operators applied to a reachable input (This is equivalent to the vertical layers definition in Figure \ref{fig:rectangularmesh}.). The reachable inputs $I_{n\ell}$ are the subset of input modes affecting the immediate inputs of the MZI at $(n,\ell)$, and the reachable outputs $O_{n\ell}$ are the subset of output modes affected by the immediate outputs of the MZI.

Following the definitions in Ref. \cite{Russell2017DirectMatrices}, in the triangular scheme, $\alpha_{n\ell}:= N-n$, and in the rectangular scheme, $\alpha_{n\ell}:= d\left(n,\ell\right)+1-s_{n\ell}[\ell]$ where $d(n,\ell)$ is the number of nodes on the diagonal (measured along paths of constant $n+\ell$) containing a rotation parameterized by $\theta_{n\ell}$, and $s_{n\ell}$ is a sequence of decreasing odd integers $d(n,\ell)\ge k_{\text{odd}}\ge1$, followed by increasing even integers $2\le k_{\text{even}}\le d(n,\ell)$, as defined in \cite{Russell2017DirectMatrices}. We prove below that for both the triangular and rectangular meshes, $\alpha_{n\ell}=\abs{I_{n\ell}}+\abs{O_{n\ell}}-N-1$.

\begin{lemma}
In the triangular mesh, $\alpha_{n\ell}=\abs{I_{n\ell}}+\abs{O_{n\ell}}-N-1$.
\end{lemma}

\begin{proof}
In the triangular mesh (shown for $N=8$ in Figure \ref{fig:triangularmesh}) $\alpha_{n\ell}:= N-n$, so we wish to show that $N-n = \abs{I_{n\ell}}+\abs{O_{n\ell}}-N-1$, or:
\begin{equation} \label{eqn:triangleinduction}
    2N+1=\abs{I_{n\ell}}+\abs{O_{n\ell}}+n.
\end{equation}

Suppose Equation \ref{eqn:triangleinduction} holds for some arbitrary $n', \ell'$ in the mesh, such that $2N+1=\abs{I_{n'\ell'}}+\abs{O_{n'\ell'}}+n'$. First, induct on $n$: if we take $n=n'+2$ and $\ell=\ell'$, then $\abs{I_{n\ell}}=\abs{I_{n'\ell'}}-1$ and $\abs{O_{n\ell}}=\abs{O_{n'\ell'}}-1$. Next, induct on $\ell$: if we take $n=n'$ and $\ell=\ell'+2$, then $\abs{I_{n\ell}}=\abs{I_{n'\ell'}}+1$ and $\abs{O_{n\ell}}=\abs{O_{n'\ell'}}-1$. In both cases, Equation \ref{eqn:triangleinduction} holds. 

Traversals by 2 along $n$ or $\ell$ from a starting node can reach all nodes with the same parity of $n$ and $\ell$, so we need two base cases. Consider the apex node at $n=1$, $\ell=N-1$ and one of its neighbors at $n=2$, $\ell=N$. The former has $\abs{I_{n\ell}}=\abs{O_{n\ell}}=N$ and the latter has $\abs{I_{n\ell}}=N$ and $\abs{O_{n\ell}}=N-1$. In both cases, Equation \ref{eqn:triangleinduction} is satisfied, so the lemma holds by induction.
\end{proof}

\begin{lemma}
In the rectangular mesh, $\alpha_{n\ell}=\abs{I_{n\ell}}+\abs{O_{n\ell}}-N-1$.
\end{lemma}

\begin{proof}
In the rectangular mesh, $\alpha_{n\ell}:= d\left(n,\ell\right)+1-s_{n\ell}[\ell]$, as defined in Ref. \cite{Russell2017DirectMatrices}. Define orthogonal axes $x$ and $y$ on the lattice such that for a node at $\left(n,\ell\right)$, traveling in the $+x$ direction gives the neighboring node at $\left(n+1,\ell+1\right)$ and traveling in the $+y$ direction gives the neighboring node at $\left(n-1,\ell+1\right)$, as depicted in Figure \ref{fig:rectangularalphamn}. For even \{odd\} $N$, let the node at $\left(n,\ell\right)=\left(1,1\right)$ have $x=1$ and the node at $\left(n,\ell\right)=\left(N-1,1\{2\}\right)$ have $y=1$. Then there is a one-to-one mapping such that $\left(x,y\right)=\left(\frac{n+\ell}{2},\frac{\ell-n}{2}+\floor{\frac{N}{2}}\right)$, as shown in Figure \ref{fig:rectangularalphamn}, and it suffices to prove the lemma by induction in this diagonal basis.

Since $d\left(n,\ell\right)$ is defined to be the length of a diagonal along paths of constant $n+\ell$, it depends only on $x$, so we rewrite $d\left(n,\ell\right)\mapsto d(x)$ explicitly:
\begin{equation}
    d(x)=\begin{cases}
        2x-1 & x\le\floor{\frac{N}{2}}\\
        2(N-x) & x>\floor{\frac{N}{2}}
    \end{cases}.
\end{equation}

Similarly, since $s_{n\ell}[\ell]$ is enumerated along a diagonal, it depends only on $y$, and we convert $s_{n\ell}[\ell]\rightarrow s_{x}[y]$ from the sequence definition of Ref. \cite{Russell2017DirectMatrices} to an explicit lattice form:
\begin{equation}
    s_{x}[y]=\begin{cases}
        2\left(\floor{\frac{N}{2}}-y\right)+1 & y\le\floor{\frac{N}{2}}\\
        2\left(y-\floor{\frac{N}{2}}\right) & y>\floor{\frac{N}{2}}
    \end{cases}.
\end{equation}
In this diagonal basis, we want to show that
\begin{equation}\label{eqn:rectangularInduction}
    d(x)+1-s_{x}[y]=\abs{I_{xy}}+\abs{O_{xy}}-N-1.
\end{equation}
There are two boundaries at $x,y=\floor{\frac{N}{2}}$ which separate four quadrants that must be considered, depicted by gray lines in Figure \ref{fig:rectangularalphamn}. We will induct on $x$ and $y$ within each quadrant, then induct on $x$ or $y$ across each of the two boundaries. 

Suppose that Equation \ref{eqn:rectangularInduction} holds for some arbitrary $x'y'$ in the mesh, such that $d\left(x'\right)+1-s_{x'}[y']=\abs{I_{x'y'}}+\abs{O_{x'y'}}-N-1$. First, we induct on $x$ and $y$ within each quadrant; the results are tabulated in Table \ref{table:quadrantinduction}. In every case, $d(x)-s_{x}[y]-\abs{I_{xy}}-\abs{O_{xy}}=d\left(n,\ell\right)-s_{x'}[y']-\abs{I_{x'y'}}-\abs{O_{x'y'}}$, so Equation \ref{eqn:rectangularInduction} remains satisfied. 

Next, we induct across the $x,y=\floor{\frac{N}{2}}$ boundaries, shown in Table \ref{table:borderInduction}. Again, in every case, $d(x)-s_{x}[y]-\abs{I_{xy}}-\abs{O_{xy}}=d\left(n,\ell\right)-s_{x'}[y']-\abs{I_{x'y'}}-\abs{O_{x'y'}}$, satisfying Equation \ref{eqn:rectangularInduction}. 

Finally, note that the base case of the top left MZI at $(n, \ell)=(1,1)$, $(x,y)=\left(1,\floor{\frac{N}{2}}\right)$ holds, with $d(x)+1-s_{x}[y]=1=2+N-N-1=\abs{I_{xy}}+\abs{O_{xy}}-N-1$. This completes the proof in the $(x,y)$ basis, and since there is a one-to-one mapping between $(x,y)\leftrightarrow(n,\ell)$, $\alpha_{n\ell}=\abs{I_{n\ell}}+\abs{O_{n\ell}}-N-1$ holds by induction.
\end{proof}

\begin{figure*}[t]
    \centering
    \includegraphics[width=0.48\textwidth]{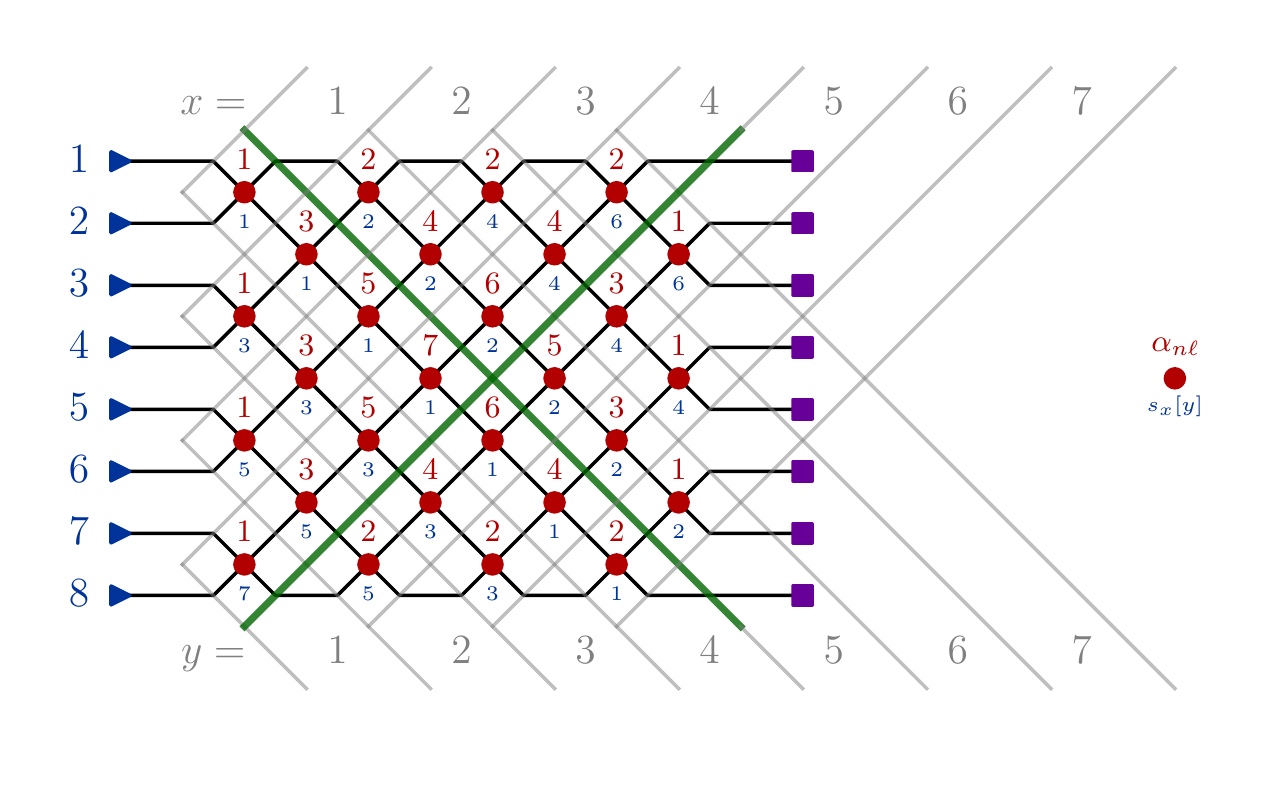}
    \includegraphics[width=0.48\textwidth]{14b_mesh_neq8_alphanl.pdf}
    \caption{Rectangular decomposition for even ($N=8$) and odd ($N=7$) meshes, showing the diagonal $x,y$ basis. Values for $\alpha_{n\ell}$ are shown in red above each MZI, with values for $s_x [y]$ shown in blue below. The critical boundaries of $x,y = \floor{\frac{N}{2}}$ separating the different quadrants are drawn in green. (Boundaries are offset for visual clarity.)}
    \label{fig:rectangularalphamn}
\end{figure*}

\begin{table*}[t]\label{table:quadrantinduction}
\centering%
\setlength{\tabcolsep}{1.0em} 
{\renewcommand{\arraystretch}{1.2}
\begin{tabular}{llllll}
\toprule 
Quadrant & Induction & $d(x)=\cdots$ & $s_{x}\left[y\right]=\cdots$ & $\abs{I_{xy}}=\cdots$ & $\abs{O_{xy}}=\cdots$ \\
\midrule
$x'\le\floor{\frac{N}{2}},y'\le\floor{\frac{N}{2}}$ & $x=x'-1$ & $d\left(n,\ell\right)-2$ & $s_{x'}[y']$ & $\abs{I_{x'y'}}-2$ & $\abs{O_{x'y'}}$ \\
 & $y=y'-1$ & $d\left(n,\ell\right)$ & $s_{x'}[y']+2$ & $\abs{I_{x'y'}}-2$ & $\abs{O_{x'y'}}$ \\
$x'\le\floor{\frac{N}{2}},y'>\floor{\frac{N}{2}}$ & $x=x'-1$ & $d\left(n,\ell\right)-2$ & $s_{x'}[y']$ & $\abs{I_{x'y'}}-2$ & $\abs{O_{x'y'}}$ \\
 & $y=y'+1$ & $d\left(n,\ell\right)$ & $s_{x'}[y']+2$ & $\abs{I_{x'y'}}$ & $\abs{O_{x'y'}}-2$ \\
$x'>\floor{\frac{N}{2}},y'\le\floor{\frac{N}{2}}$ & $x=x'+1$ & $d\left(n,\ell\right)-2$ & $s_{x'}[y']$ & $\abs{I_{x'y'}}$ & $\abs{O_{x'y'}}-2$ \\
 & $y=y'-1$ & $d\left(n,\ell\right)$ & $s_{x'}[y']+2$ & $\abs{I_{x'y'}}-2$ & $\abs{O_{x'y'}}$ \\
$x'>\floor{\frac{N}{2}},y'>\floor{\frac{N}{2}}$ & $x=x'+1$ & $d\left(n,\ell\right)-2$ & $s_{x'}[y']$ & $\abs{I_{x'y'}}$ & $\abs{O_{x'y'}}-2$ \\
 & $y=y'+1$ & $d\left(n,\ell\right)$ & $s_{x'}[y']+2$ & $\abs{I_{x'y'}}$ & $\abs{O_{x'y'}}-2$ \\
\bottomrule
\end{tabular}
}
\caption{Induction on $x$ and $y$ within each of the quadrants in the mesh.}
\end{table*}

\begin{table*}[t] \label{table:borderInduction}
\centering%
\setlength{\tabcolsep}{1.0em} 
{\renewcommand{\arraystretch}{1.2}
\begin{tabular}{lllllll}
\toprule 
$x'$ & $y'$ & Induction & $d(x)=\cdots$ & $s_{x}\left[y\right]=\cdots$ & $\abs{I_{xy}}=\cdots$ & $\abs{O_{xy}}=\cdots$ \\
\midrule
$x'=\floor{\frac{N}{2}}$ & any & $x=x'+1$ & $d\left(n, \ell\right)-\{+\}1$ & $s_{x'}[y']$ & $\abs{I_{x'y'}}+0\{1\}$ & $\abs{O_{x'y'}}-1\{0\}$ \\
any & $y'=\floor{\frac{N}{2}}$ & $y=y'+1$ & $d\left(n, \ell\right)$ & $s_{x'}[y']+1$ & $\abs{I_{x'y'}}$ & $\abs{O_{x'y'}}-1$ \\
\bottomrule
\end{tabular}
}
\caption{Induction on $x$ or $y$ across each of the borders of $x,y=\protect\floor{\frac{N}{2}}$.}
\end{table*}

\end{document}